\documentclass[11pt]{article}

\usepackage{amssymb}
\usepackage{latexsym}
\usepackage{amsmath}
\usepackage{euscript}
\usepackage{graphics}
\usepackage{color}
\usepackage{pdfcolmk}

\def\bm\chi{\mbox{\boldmath$\chi$}}

\let\xker=\ker \def\ker{{\xker\,}}

\newtheorem{theorem}{Theorem}[section]
\newtheorem{proposition}[theorem]{Proposition}
\newtheorem{corollary}[theorem]{Corollary}
\newtheorem{lemma}[theorem]{Lemma}
\newtheorem{definition}[theorem]{Definition}

\newtheorem{remark}[theorem]{Remark}

\numberwithin{equation}{section}
\newcommand{\ba}{\begin{array}}
\newcommand{\ea}{\end{array}}
\newcommand{\bad}{\begin{array*}}
\newcommand{\ead}{\end{array*}}
\newcommand{\bea}{\begin{eqnarray}}
\newcommand{\eea}{\end{eqnarray}}
\newcommand{\bead}{\begin{eqnarray*}}
\newcommand{\eead}{\end{eqnarray*}}
\newcommand{\be}{\begin{equation}}
\newcommand{\ee}{\end{equation}}
\newcommand{\bed}{\begin{displaymath}}
\newcommand{\eed}{\end{displaymath}}
\newcommand{\bl}{\begin{lemma}}
\newcommand{\el}{\end{lemma}}
\newcommand{\bp}{\begin{proposition}}
\newcommand{\ep}{\end{propostion}}
\newcommand{\bt}{\begin{theorem}}
\newcommand{\et}{\end{theorem}}

\newcommand{\bc}{\begin{corollary}}
\newcommand{\ec}{\end{corollary}}

\newcommand{\br}{\begin{remark}}
\newcommand{\er}{\end{remark}}
\newcommand{\bd}{\begin{definition}}
\newcommand{\ed}{\end{definition}}

\newenvironment{proof}%
{\begin{sloppypar}\noindent{\bf Proof.}}%
{\hspace*{\fill}$\square$\end{sloppypar}}

\newenvironment{proof1}%
{\begin{sloppypar}\noindent{\bf Proof of Theorem $\ref{THM 2}$.}}%
{\hspace*{\fill}$\square$\end{sloppypar}}

\newenvironment{proof2}%
{\begin{sloppypar}\noindent{\bf Proof of Theorem $\ref{THM 3}$.}}%
{\hspace*{\fill}$\square$\end{sloppypar}}

\title{\bf On the spectrum of the lattice spin-boson Hamiltonian for any coupling: 1D case}

\author{Mukhiddin~Muminov\\
Faculty of Scinces, Universiti Teknologi Malaysia (UTM)\\
81310 Skudai, Johor Bahru, Malaysia\\
E-mail: mmuminov@mail.ru \and
Hagen~Neidhardt\\
Weierstrass Institute for Applied Analysis and  Stochastics\\
Mohrenstr. 39, D-10117 Berlin, Germany\\
E-mail: neidhard@wias-berlin.de \and
Tulkin~Rasulov\\
Faculty of Physics and Mathematics, Bukhara State University\\
M. Ikbol str. 11, 200100 Bukhara, Uzbekistan\\
E-mail: rth@mail.ru}

\begin{document}

\maketitle

\vspace{-7mm}

\vspace{-5mm}

\noindent {\bf Abstract:}
A lattice model of radiative decay (so-called spin-boson model) of a two level atom and at most
two photons is considered. The location of the essential spectrum is described.
For any coupling constant the finiteness
of the number of eigenvalues below the bottom of its essential spectrum is proved.
The results are obtained by considering a more general model $H$ for which
the lower bound of its essential spectrum is estimated. Conditions which guarantee
the finiteness of the number of eigenvalues of $H,$
below the bottom of its essential spectrum are found.
It is shown that the discrete spectrum might be infinite if
the parameter functions are chosen in a special form.

\section{Introduction}

Block operator matrices are matrices where the entries are
linear operators between Banach or Hilbert spaces \cite{CT08}. One
special class of block operator matrices are Hamiltonians
associated with systems of non-conserved number of quasi-particles
on a lattice. Their number can be unbounded as in the case of spin-boson models
or bounded as in the case of "truncated" $ $ spin-boson models.
They arise, for example, in the theory of solid-state physics \cite{Mog},
quantum field theory \cite{Frid} and statistical physics
\cite{Mal-Min, MS}.

In a well-known model of radiative decay (the so-called spin-boson model)
it is assumed that an atom, which can be in two states -- ground state with energy
$-\varepsilon$ and excited state with energy $\varepsilon$ -- emits and absorbs photons,
going over from one state to the other \cite{HuebnerSpohn, MS, Spohn89, ZhM}.
The energy operator of such a system is given
by the (formal) expression \cite{HuebnerSpohn, MS, Spohn89, ZhM}
\begin{eqnarray}
{\cal A}:=\varepsilon \sigma_z+\int_{{\Bbb R}^{\rm d}} w(k)a^*(k)a(k)dk+\alpha \sigma_x
\int_{{\Bbb R}^{\rm d}} v(k)(a^*(k)+a(k))dk
\label{spin-boson Hamiltonian}
\end{eqnarray}
and acts in the Hilbert space
\begin{eqnarray}
{\cal L}:={\Bbb C}^2 \otimes {\mathcal F}_{\rm s}(L_2({\Bbb R}^{\rm d})),
\label{Eq02}
\end{eqnarray}
where ${\Bbb C}^2$ is the state of the two-level atom and ${\mathcal F}_{\rm s}(L_2({\Bbb R}^{\rm d}))$
is the symmetric Fock space for bosons. In the following we consider the lattice analog
of the standard spin-boson Hamiltonian cf.~\cite{Mog}.
In the "algebraic" sense, a lattice spin-boson Hamiltonian is similar to a
standard one with only the difference is that ${\cal A}$ does not act
in the Euclidean space ${\Bbb R}^{\rm d}$ but on a ${\rm d}$--dimensional torus ${\Bbb T}^{\rm d}.$
This means that we have to replace ${\Bbb R}^{\rm d}$ by ${\Bbb T}^{\rm d}$
in formulas (\ref{spin-boson Hamiltonian}) and (\ref{Eq02}).
We write elements $F$ of the space ${\cal L}$ in the form
$$
F=\{f_0^{(\sigma)},f_1^{(\sigma)}(k_1),f_2^{(\sigma)}(k_1,k_2),\ldots,f_n^{(\sigma)}(k_1,k_2,\ldots,k_n),\ldots\}
$$
of functions of an increasing number of variables $(k_1,\ldots,k_n),$ $k_i \in {\Bbb T}^{\rm d},$ and
a discrete variable $\sigma=\pm;$ the functions are symmetric with respect to the variables $k_i,$
$i=1,\ldots,n,$ $n \in {\Bbb N}.$ The norm in ${\mathcal L}$ is given by
\begin{eqnarray}
\|F\|^2:=\sum\limits_{\sigma=\pm}|f_0^{(\sigma)}|^2+\sum\limits_{\sigma, n} \int_{({\Bbb T}^{\rm d})^n}
|f_n^{(\sigma)}(k_1,\ldots,k_n)|^2 dk_1 \ldots dk_n.
\label{norm}
\end{eqnarray}

In the expression (\ref{spin-boson Hamiltonian}), the operators $a^*(k)$ and $a(k)$ are "creation and annihilation"
$ $ operators, $\varepsilon>0,$
$$
\sigma_z:=\left( \begin{array}{cc}
1 & 0\\
0 & -1\\
\end{array}
\right), \quad \sigma_x:=\left( \begin{array}{cc}
0 & 1\\
1 & 0\\
\end{array}
\right)
$$
are Pauli matrices, $w(k)$ is the dispersion of the free field,
$\alpha v(k)$ is the coupling between the atoms and
the field modes, $\alpha>0$ is the coupling constant.

However, the problem of complete spectral description of the operator ${\cal A}$ still seems
rather difficult. In this connection, it is natural to consider simplified ("truncated") models
\cite{HuebnerSpohn,MS,ZhM} that
differ from the model described above with respect to the number of bosons which is bounded
by $N,$ $N \in {\Bbb N}.$ The Hilbert state space of each such model is then the space
${\mathcal L}_{N}:={\Bbb C}^2 \otimes {\mathcal F}_{\rm s}^{(N)}(L_2({\Bbb T}^{\rm d})) \subset {\mathcal L},$
where
$$
{\mathcal F}_{\rm s}^{(N)}(L_2({\Bbb T})):={\Bbb C} \oplus L_2({\Bbb T}^{\rm d}) \oplus L_2^{\rm sym} (({\Bbb
T}^{\rm d})^2) \oplus \ldots \oplus L_2^{\rm sym} (({\Bbb T}^{\rm d})^N).
$$
Here $L_2^{\rm sym} (({\Bbb T}^{\rm d})^n)$ is the Hilbert space of symmetric functions of $n$
variables, and the norm in ${\mathcal L}_N$ is introduced as in (\ref{norm}).
Then the truncated Hamiltonian ${\cal A}_N$ is given in ${\mathcal L}_N$ by
${\cal A}_N:=P_{{\mathcal L}_N} {\cal A} P_{{\mathcal L}_N},$
where $P_{{\mathcal L}_N}$ is the projection of the space ${\mathcal L}$ onto the subspace ${\mathcal L}_N,$
and ${\cal A}$ is the Hamiltonian (\ref{spin-boson Hamiltonian}).

The standard spin-boson Hamiltonian with $N=1,2$ was completely studied
in \cite{MS} for small values of the parameter $\alpha.$ The case $N=3$ was considered
in \cite{ZhM}. The existence of wave operators and their asymptotic completeness were
proven there. In \cite{HuebnerSpohn}, the case of arbitrary $N$ was investigated.
In particular, using a Mourre
type estimate, a complete spectral characterization of the spin-boson Haniltonian  are given
for sufficiently small, but nonzero coupling constant.

Let us introduce the corresponding model operator for the case $N=2.$
For simplicity we denote ${\mathcal H}_0:={\Bbb C},$ ${\mathcal H}_1:=L_2({\Bbb
T}^{\rm d}),$ ${\mathcal H}_2:=L_2^{\rm sym}(({\Bbb T}^{\rm d})^2)$ and
${\mathcal H}:={\mathcal F}_{\rm s}^{(2)}(L_2({\Bbb T}^{\rm d})).$

In the Hilbert space ${\mathcal H}$ we consider the model operator $H$
that admit an $3 \times 3$ tridiagonal block operator matrix
representation
$$
H:=\left( \begin{array}{ccc}
H_{00} & H_{01} & 0\\
H_{01}^* & H_{11} & H_{12}\\
0 & H_{12}^* & H_{22}\\
\end{array}
\right)
$$
with the entries $H_{ij}: {\mathcal H}_j \to {\mathcal H}_i,$ $i
\le j,$ $i,j=0,1,2$ defined by
\begin{eqnarray*}
&&H_{00}f_0=w_0f_0, \quad H_{01}f_1=\int_{{\Bbb T}^{\rm d}}
v_0(t)f_1(t)dt,\\
&& (H_{11}f_1)(x)=w_1(x)f_1(x),\quad (H_{12}f_2)(x)= \int_{{\Bbb T}^{\rm d}} v_1(t) f_2(x,t)dt,\\
&& (H_{22}f_2)(x,y)=w_2(x,y)f_2(x,y), \,\, f_i \in {\mathcal H}_i, \,\, i=0,1,2;
\end{eqnarray*}
where $w_0$ is a real number,
$w_1(\cdot)$ $v_i(\cdot),$ $i=0,1,$ are real-valued analytic
functions on ${\Bbb T}^{\rm d}$ and $w_2(\cdot, \cdot)$ is a real-valued
symmetric analytic function on $({\Bbb T}^{\rm d})^2.$
Under these assumptions the operator $H$ is bounded and
self-adjoint.

An important problem in the spectral theory of such
model operators is to study the number of eigenvalues located outside the essential spectrum.
We remark that the operator $H$ has been considered before and the following results were obtained:
The location of the essential spectrum of $H$ has been described in \cite{LR03} for any ${\rm d} \geq 1.$
The existence of infinitely many eigenvalues below the bottom of the essential spectrum of $H$ has been
announced in \cite{LR03-1} for ${\rm d}=3.$ Its complete proof was given in \cite{ALR1} for a
special parameter functions and in \cite{ALR} for more general case. An asymptotics of the form
${\mathcal U}_0|\log|\lambda||$ ($0<{\mathcal U}_0<\infty$) for the number of eigenvalues on the left
of $\lambda,$ $\lambda<\min\sigma_{\rm ess}(H)$ was obtained in \cite{ALR}. The conditions for the finiteness
of the discrete spectrum of $H$ was found in \cite{Ras07} for the case ${\rm d}=3.$

In the present paper we consider the case ${\rm d}=1.$
We study the relation between the lower
bounds of the two-particle and three-particle branches of the
essential spectrum of $H.$ Under natural assumptions on the
parameters we prove the
finiteness of the discrete spectrum of $H$ using the Birman-Schwinger principle.
In this analysis the finiteness
of the points which gives global minima for the function $w_2(\cdot,\cdot)$
is important.
We give a counter-example, which shows the infiniteness of the discrete spectrum if
the number of such points is not finite.
For this case the exact view of the eigenvalues
and eigenvectors are found, their multiplicities are calculated. We notice that a part
of the results is typical for ${\rm d}=1,$
in fact, they do not have analogues in the case ${\rm d} \geq 2.$

Using a connection between the operators ${\mathcal A}_2$ and $H,$
and applying obtained results from above we describe the essential spectrum of ${\mathcal A}_2,$
and show that the operator ${\cal A}_2$ has finitely many eigenvalues below the bottom of its essential spectrum
for {\it any} coupling constant $\alpha.$

The paper is organized as follows. Section 1 is
an introduction. In Section 2, the main results for $H$
are formulated. In Section 3, we estimate the
lower bound of the essential spectrum of $H.$ In Section 4, we
apply the Birman-Schwinger principle to $H.$
Section 5 is devoted to the proof of the finiteness of the
discrete spectrum of $H.$ In Section 6 we discuss the case when
the discrete spectrum of $H$ is infinite. In Section 7,
the finiteness and infiniteness of the discrete spectrum of $H$
is established, when the lower bounds of the two- and three-particle
branches of the essential spectrum are coincide. An application
to lattice model of radiative decay of a two level atom and at most
two photons (truncated spin-boson model on a lattice) illustrates our results.

\section{The main results for $H$}

The spectrum, the essential spectrum, the point spectrum and the discrete spectrum of
a bounded self-adjoint operator will be denoted by
$\sigma(\cdot),$ $\sigma_{\rm ess}(\cdot),$ $\sigma_{\rm p}(\cdot)$ and $\sigma_{\rm
disc}(\cdot),$ respectively.

To study the spectral properties of $H$ we introduce a following
family of bounded self-adjoint operators (generalized Friedrichs
models) $h(x),$ $x \in {\Bbb T},$ which acts in ${\mathcal H}_0
\oplus {\mathcal H}_1$ as
$$
h(x):=\left( \begin{array}{cc}
h_{00}(x) & h_{01}\\
h_{01}^* & h_{11}(x)\\
\end{array}
\right),
$$
where
\begin{eqnarray*}
&& h_{00}(x)f_0=w_1(x)f_0,\quad h_{01}f_1=\frac{1}{\sqrt{2}}
\int_{\Bbb T} v_1(t) f_1(t)dt,\\
&& (h_{11}(x)f_1)(y)=w_2(x,y)f_1(y),\quad f_i \in {\cal H}_i,\quad i=0,1.
\end{eqnarray*}

Let the operator $h_0(x),$ $x \in {\Bbb T}$ act in ${\mathcal H}_0
\oplus {\mathcal H}_1$ as
$$
h_0(x):=\left( \begin{array}{cc}
0 & 0\\
0 & h_{11}(x)\\
\end{array}
\right).
$$
The perturbation $h(x)-h_0(x)$ of the operator $h_0(x)$ is a
self-adjoint operator of rank 2. Therefore in accordance with the
Weyl theorem about the invariance of the essential spectrum under
the finite rank perturbations, the essential spectrum of the
operator $h(x)$ coincides with the essential spectrum of $h_0(x).$
It is evident that $\sigma_{\rm ess}(h_0(x))=[m_x, M_x],$ where
the numbers $m_x$ and $M_x$ are defined by
$$
m_x:= \min_{y \in {\Bbb T}} w_2(x,y) \quad \mbox{and} \quad M_x:=
\max_{y \in {\Bbb T}} w_2(x,y).
$$
This yields $\sigma_{\rm ess}(h(x))=[m_x, M_x].$

For any $x\in {\Bbb T}$ we define an analytic function
$\Delta(x\,; \cdot)$ (the Fredholm determinant associated with the
operator $h(x)$) in ${\Bbb C} \setminus [m_x, M_x]$ by
$$
\Delta(x\,; z):=w_1(x)-z-\frac{1}{2} \int_{\Bbb T}
\frac{v_1^2(t)dt}{w_2(x,t)-z}.
$$

Note that for the discrete spectrum of $h(x)$ the equality
$$
\sigma_{\rm disc}(h(x))=\{z \in {\Bbb C} \setminus [m_x, M_x]:\,
\Delta(x\,; z)=0 \}
$$
holds (see Lemma \ref{LEM 1}).

The following theorem \cite{LR03} describes the location of
the essential spectrum of $H$ by the spectrum of the family $h(x)$
of generalized Friedrichs model.

\begin{theorem}\label{THM 1} For the essential spectrum of $H$
the following equality holds
$$
\sigma_{\rm ess}(H)=\sigma \cup [m, M], \quad
\sigma:=\bigcup\limits_{x \in {\Bbb T}} \sigma_{\rm disc}(h(x))
$$
where the numbers $m$ and $M$ are defined by
$$
m:= \min\limits_{x,y\in {\Bbb T}} w_2(x,y) \quad \mbox{and} \quad
M:= \max\limits_{x,y\in {\Bbb T}} w_2(x,y).
$$
\end{theorem}

The sets $\sigma$ and $[m, M]$ are called two- and three-particle
branches of the essential spectrum of $H,$ respectively.

Throughout this paper we assume that the function
$w_2(\cdot,\cdot)$ has a unique non-degenerate global minimum at the
point $(0,0) \in {\Bbb T}^2.$

For $\delta>0$ and $a \in {\Bbb T}$ we set
$$
U_\delta(a):=\{x \in {\Bbb T}: |x-a|<\delta\}.
$$

We remark that if $v_1(0)=0,$ then from analyticity of
$v_1(\cdot)$ on ${\Bbb T}$ it follows that there exist positive
numbers $C_1, C_2$ and $\delta$ such that the inequalities
\begin{eqnarray}
C_1 |x|^\alpha \le |v_1(x)| \le C_2 |x|^\alpha, \quad x \in
U_\delta(0)
\label{Eq1}
\end{eqnarray}
hold for some $\alpha \in {\Bbb N}.$ Since the function
$w_2(0,\cdot)$ has a unique non-degenerate global minimum at $y=0$ (see
proof of Lemma \ref{LEM 2}), one
can easily seen from the estimate \eqref{Eq1} that for any $x \in {\Bbb T}$ the integral
$$
\int_{\Bbb T} \frac{v_1^2(t)dt}{w_2(x,t)-m}
$$
is positive and finite. The Lebesgue dominated convergence theorem
yields
$$
\Delta(0\,; m)=\lim\limits_{x\to 0} \Delta(x\,; m),
$$
and hence if $v_1(0)=0,$ then the function $\Delta(\cdot\,; m)$ is
continuous on ${\Bbb T}.$

Let us denote by $E_{\rm min}$ the lower bound of the essential
spectrum of $H.$

The main results of the present paper as follows.

\begin{theorem}\label{THM 2}
For the lower bound $E_{\rm min}$ the following assertions hold:\\
{\rm (i)} If $v_1(0) \neq 0,$ then $E_{\rm min}<m;$\\
{\rm (ii)} If $v_1(0)=0$ and $\min\limits_{x \in
{\Bbb T}} \Delta(x\,; m) < 0,$ then $E_{\rm min}<m;$\\
{\rm (iii)} If $v_1(0)=0$ and $\min\limits_{x \in {\Bbb T}}
\Delta(x\,; m) \ge 0,$ then $E_{\rm min}=m.$
\end{theorem}

\begin{theorem}\label{THM 3}
If one of the assertions\\
{\rm (i)} $v_1(0) \neq 0;$\\
{\rm (ii)} $v_1(0)=0$ and $\min\limits_{x \in
{\Bbb T}} \Delta(x\,; m) < 0;$\\
{\rm (iii)} $v_1(0)=0$ and $\min\limits_{x \in {\Bbb T}}
\Delta(x\,; m) > 0,$\\
is satisfied, then the operator $H$ has a finite number of eigenvalues lying
below $E_{\rm min}.$
\end{theorem}

\begin{remark}
Since the function $w_2(\cdot,\cdot)$ is continuous on the
compact set ${\Bbb T}^2$ there exist at least one point $(x_0,y_0)
\in {\Bbb T}^2$ such that the function $w_2(\cdot,\cdot)$ attains
its global maximum at this point. If $v_1(y_0) \neq 0,$ then
similar arguments show that for the upper bound $E_{\rm max}$ of
the essential spectrum of $H$ we have $E_{\rm max}>M$ and the
operator $H$ has a finite number of eigenvalues greater than
$E_{\rm max}.$
\end{remark}

\begin{remark}
The results can be easily generalized to the case when the function $w_2(\cdot,\cdot)$
has a finite number of non-degenerate global minima at several points of ${\Bbb T}^2.$
Here finiteness of the number of such points is important.
If the number of such points is infinite, then the discrete spectrum of $H$
can be infinite, for a corresponding example see Section $6.$
\end{remark}

\begin{remark}
The case $v_1(0)=0$ and $\min\limits_{x \in {\Bbb T}}
\Delta(x\,; m)=0$ is considered in Section~$7,$ where it is shown that the discrete spectrum
of $H$ can be finite or infinite depending on the parameter functions.
\end{remark}

\section{Lower bound of the essential spectrum of $H$}

In this section first we study the discrete spectrum of $h(x)$ and
then Theorem~\ref{THM 2} will be proven.

\begin{proposition}\label{Perturbation determinant}
The perturbation determinant $\Delta_{h(x)/h_0(x)}(z)$ of the operator $h_0(x)$
by the operator $h(x)-h_0(x)$ has form
$$
\Delta_{h(x)/h_0(x)}(z)=-\frac{1}{z} \Delta(x\,; z), \quad z \in {\Bbb C} \setminus \sigma(h_0(x)).
$$
\end{proposition}

\begin{proof}
Since the operator $h(x)-h_0(x)$ is trace class, even of rank 2, the perturbation determinant $\Delta_{h(x)/h_0(x)}(z)$
is well-defined by
$$
\Delta_{h(x)/h_0(x)}(z):={\rm det} \left( I+(h(x)-h_0(x))(h_0(x)-z)^{-1} \right).
$$
Without loss of generality we can assume that $\|v_1\|=1.$ We choose the orthonormal basis $\{\varphi_n\}_n \subset {\cal H}_1$
by the following way: $\varphi_1:=v_1$ and $\varphi_j \bot v_1$ for all $j \geq 2.$
Introduce
$$
\psi_1:=\frac{1}{\sqrt{2}} \left( \begin{array}{cc}
1\\
\varphi_1\\
\end{array}
\right), \quad \psi_2:=\frac{1}{\sqrt{2}} \left( \begin{array}{cc}
1\\
-\varphi_1\\
\end{array}
\right), \quad \psi_j:=\left( \begin{array}{cc}
0\\
\varphi_{j-1}\\
\end{array}
\right), \quad j \geq 3.
$$
By the construction $\{\psi_n\}_n \subset {\cal H}_0 \oplus {\cal H}_1$ is an orthonormal.
Set
$$
a_{ij}(x\,; z):=((h(x)-h_0(x))(h_0(x)-z)^{-1} \psi_i, \psi_j),\quad i,j \in {\Bbb N}.
$$
Simple calculation show that
\begin{eqnarray*}
&& a_{11}(x\,; z)=-\frac{1}{z}w_1(x)+\frac{1}{\sqrt{2}} \int_{\Bbb T} \frac{v_1^2(t)dt}{w_2(x,t)-z}-\frac{1}{z}\frac{1}{\sqrt{2}};\\
&& a_{12}(x\,; z)=-\frac{1}{z}w_1(x)+\frac{1}{\sqrt{2}} \int_{\Bbb T} \frac{v_1^2(t)dt}{w_2(x,t)-z}+\frac{1}{z}\frac{1}{\sqrt{2}};\\
&& a_{21}(x\,; z)=-\frac{1}{z}w_1(x)-\frac{1}{\sqrt{2}} \int_{\Bbb T} \frac{v_1^2(t)dt}{w_2(x,t)-z}-\frac{1}{z}\frac{1}{\sqrt{2}};\\
&& a_{22}(x\,; z)=-\frac{1}{z}w_1(x)-\frac{1}{\sqrt{2}} \int_{\Bbb T} \frac{v_1^2(t)dt}{w_2(x,t)-z}+\frac{1}{z}\frac{1}{\sqrt{2}};\\
&& a_{ij}(x\,; z)=\delta_{ij}, \quad {\rm otherwise}.
\end{eqnarray*}
Here $\delta_{ij}$ is Kronecker delta. Therefore,
$$
\Delta_{h(x)/h_0(x)}(z)=\frac{1}{4} {\rm det} \left( \begin{array}{cc}
2-a_{11}(x\,; z) & a_{12}(x\,; z)\\
a_{21}(x\,; z) & 2-a_{22}(x\,; z)\\
\end{array}
\right)=-\frac{1}{z} \Delta(x\,; z).
$$
Proposition is proved.
\end{proof}

The following lemma is a simple consequence of the
Proposition~\ref{Perturbation determinant} and of \cite[chapter~IV]{Gohberg-Krein}.

\begin{lemma}\label{LEM 1} For any fixed $x \in {\Bbb T}$ the operator $h(x)$ has an eigenvalue
$z(x) \in {\Bbb C} \setminus [m_x, M_x]$ if and only if
$\Delta(x\,; z(x))=0.$
\end{lemma}

In the next two lemmas we describe the number and location of the
eigenvalues of $h(x).$

\begin{lemma}\label{LEM 2}
If $v_1(0) \neq 0,$ then there exists $\delta>0$ such that for any
$x \in U_\delta(0)$ the operator $h(x)$ has a unique eigenvalue
$z(x),$ lying on the left of $m_x.$
\end{lemma}

\begin{proof}
Since the function $w_2(\cdot,\cdot)$ has a unique non-degenerate global
minimum at the point $(0,0) \in {\Bbb T}^2,$ by the implicit
function theorem there exist $\delta>0$ and an analytic function
$y_0(\cdot)$ on $U_\delta(0)$ such that for any $x \in
U_\delta(0)$ the point $y_0(x)$ is the unique non-degenerate
minimum of the function $w_2(x,\cdot)$ and $y_0(0)=0.$ Therefore,
we have $w_2(x,y_0(x))=m_x$ for any $x \in U_\delta(0).$

Let $\widetilde{w}_2(\cdot,\cdot)$ be the function defined on
$U_\delta(0) \times {\Bbb T}$ as
$\widetilde{w}_2(x,y):=w_2(x,y+y_0(x))-m_x.$
Then for any $x \in U_\delta(0)$ the function
$\widetilde{w}_2(x,\cdot)$ has a unique non-degenerate zero
minimum at the point $0 \in {\Bbb T}.$ Now using the equality
$$
\int_{\Bbb T} \frac{v_1^2(t)dt}{w_2(x,t)-m_x}=
\int_{\Bbb T} \frac{v_1^2(t+y_0(x))dt}
{\widetilde{w}_2(x,t)}, \quad x \in U_\delta(0),
$$
the continuity of the function $v_1(\cdot),$ the conditions
$v_1(0) \neq 0$ and $y_0(0)=0$ it is easy to see that
$
\lim\limits_{z \to m_x-0} \Delta(x\,; z)=-\infty
$
for all $x \in U_\delta(0).$

Since for any $x \in {\Bbb T}$ the function $\Delta(x\,; \cdot)$
is continuous and monotonically decreasing on $(-\infty, m_x)$ the
equality
\begin{eqnarray}
\lim\limits_{z \to -\infty} \Delta(x\,; z)=\infty
\label{Eq2}
\end{eqnarray}
implies that for any $x \in U_\delta(0)$ the function $\Delta(x\,;
\cdot)$ has a unique zero $z=z(x),$ lying in $(-\infty, m_x).$ By
Lemma \ref{LEM 1} the number $z(x)$ is the eigenvalue of $h(x).$
\end{proof}

\begin{lemma}\label{LEM 6}
Let $v_1(0) = 0.$\\
{\rm (i)} If $\min\limits_{x \in {\Bbb T}} \Delta(x\,; m) \ge 0,$
then for any $x \in {\Bbb T}$ the operator $h(x)$ has no
eigenvalues, lying on the left of $m.$\\
{\rm (ii)} If $\min\limits_{x \in {\Bbb T}} \Delta(x\,; m) < 0,$
then there exists a non-empty set $G \subset {\Bbb T}$ such that
for any $x \in G$ the operator $h(x)$ has a unique eigenvalue
$z(x),$ lying on the left of $m.$
\end{lemma}

\begin{proof}
First we recall that if $v_1(0) = 0,$ then the function
$\Delta(\cdot\,;m)$ is a continuous on ${\Bbb T}.$ Let
$\min\limits_{x \in {\Bbb T}} \Delta(x\,; m) \ge 0.$ Since for any
$x \in {\Bbb T}$ the function $\Delta(x\,; \cdot)$ is
monotonically decreasing on $(-\infty, m)$ we have
$
\Delta(x\,; z) > \Delta(x\,; m) \geq \min\limits_{x \in {\Bbb T}}
\Delta(x\,; m) \ge 0,
$
that is, $\Delta(x\,; z)>0$ for all $x \in {\Bbb T}$ and $z<m.$
Therefore, by Lemma~\ref{LEM 1} for any $x \in {\Bbb T}$ the
operator $h(x)$ has no eigenvalues in $(-\infty, m).$

Now we suppose that $\min\limits_{x \in {\Bbb T}} \Delta(x\,; m) <
0$ and introduce the following subset of ${\Bbb T}:$
$$
G:= \{ x \in {\Bbb T}: \Delta(x\,; m)<0\}.
$$

Since $\Delta(\cdot\,; m)$ is a continuous on the compact set
${\Bbb T},$ there exists a point $x^0 \in{\Bbb T}$ such that
$
\min\limits_{x \in {\Bbb T}} \Delta(x\,; m)=\Delta(x^0\,; m),
$
that is, $x^0 \in G.$ So, the set $G$ is a non-empty. Note that if
$\max\limits_{x \in {\Bbb T}} \Delta(x\,; m)<0,$ then
$\Delta(x\,; m)<0$ for all $x \in {\Bbb T}$ and hence $G={\Bbb T}.$

Since for any $x\in {\Bbb T}$ the function $\Delta(x\,; \cdot)$ is
a continuous and monotonically decreasing on $(-\infty, m]$ by the
equality \eqref{Eq2} for any $x \in G$ there exists a unique point
$z(x) \in (-\infty, m)$ such that $\Delta(x\,; z(x))=0.$ By Lemma~\ref{LEM 1}
for any $x \in G$ the point $z(x)$ is the unique
eigenvalue of $h(x).$

By the construction of $G$ the inequality $\Delta(x\,; m) \ge 0$
holds for all $x \in {\Bbb T} \setminus G.$ In this case one can
see that for any $x \in {\Bbb T} \setminus G$ the operator $h(x)$
has no eigenvalues in $(-\infty, m).$
\end{proof}

\begin{proof1} Let $v_1(0) \neq 0.$
Then by Lemma~\ref{LEM 2} there exists $\delta>0$ such that for
any $x \in U_\delta(0)$ the operator $h(x)$ has a unique
eigenvalue $z(x),$ lying on the left of $m_x.$ In particular,
$z(0)<m_0.$ Since $m=\min\limits_{x \in {\Bbb T}} m_x=m_0$ it
follows that $\min \sigma \le z(0) < m,$ that is, $E_{\rm min}<m.$

Let $v_1(0)=0.$ Then two cases are possible: $\min\limits_{x \in
{\Bbb T}} \Delta(x\,; m) \ge 0$ or $\min\limits_{x \in {\Bbb T}}
\Delta(x\,; m) < 0.$ In the case $\min\limits_{x \in {\Bbb T}}
\Delta(x\,; m) \ge 0,$ by the part (i) of Lemma~\ref{LEM 6} for any
$x \in {\Bbb T}$ the operator $h(x)$ has no eigenvalues in
$(-\infty, m),$ that is, $\min\sigma \ge m.$ By Theorem~\ref{THM
1} it means that $E_{\rm min}=m.$

For the case $\min\limits_{x \in {\Bbb T}} \Delta(x\,; m) < 0,$
using the part (ii) of Lemma~\ref{LEM 6} we obtain $\min \sigma
\le z(x') < m$ for all $x' \in G,$ that is, $E_{\rm min}<m.$
\end{proof1}

\section{The Birman-Schwinger principle.}

For a bounded self-adjoint operator $A$ acting in the Hilbert
space ${\mathcal R}$ and for a real number $\lambda,$ we define \cite{Glazman}
the number $n(\lambda, A)$ by the rule
$$
n(\lambda, A):=\sup \{{\rm dim}(F): (Au,u)>\lambda,\, u\in F \subset
{\mathcal R},\,||u||=1\}.
$$

The number $n(\lambda, A)$ is equal to infinity if
$\lambda<\max\sigma_{\rm ess}(A);$ if $n(\lambda, A)$ is finite,
then it is equal to the number of the eigenvalues of $A$ bigger
than $\lambda.$

Let us denote by $N(z)$ the number of eigenvalues of $H$ on the
left of $z,$ $z \leq E_{\rm min}.$ Then we have
$
N(z)=n(-z, -H),\quad -z>-E_{\rm min}.
$

Since the function $\Delta(\cdot\,; \cdot)$ is positive on
$(x,z) \in {\Bbb T} \times (-\infty, E_{\rm min}),$ there exists
a positive square root of $\Delta(x\,; z)$ for all $x \in {\Bbb T}$
and $z < E_{\rm min}.$

In our analysis of the discrete spectrum of $H$ the crucial role
is played by the self-adjoint compact $2 \times 2$ block operator
matrix $T(z),$ $z < E_{\rm min}$ acting on ${\mathcal H}_0 \oplus
{\mathcal H}_1$ as
$$
T(z):=\left( \begin{array}{cc}
T_{00}(z) & T_{01}(z)\\
T_{01}^*(z) & T_{11}(z)\\
\end{array}
\right)
$$
with the entries $T_{ij}(z): {\mathcal H}_j \to {\mathcal H}_i,$
$i \le j,$ $i,j=0,1$ defined by
\begin{eqnarray*}
&&T_{00}(z)g_0=(1+z-w_0)g_0,\quad T_{01}(z)g_1=
-\int_{\Bbb T} \frac{v_0(t)g_1(t)dt}{\sqrt{\Delta(t\,; z)}};\\
&&(T_{11}(z)g_1)(x)=\frac{v_1(x)}{2\sqrt{\Delta(x\,; z)}}
\int_{\Bbb T} \frac{v_1(t)g_1(t)dt} {\sqrt{\Delta(t\,; z)}(w_2(x,t)-z)}.
\end{eqnarray*}
Here $g_i \in {\mathcal H}_i,$ $i=0,1.$

The following lemma is a modification of the well-known
Birman-Schwinger principle for the operator $H$ (see \cite{ALR}).

\begin{lemma}\label{LEM 3}
The operator $T(z)$ is compact and continuous in $z < E_{\rm min}$
and
$$
N(z) = n(1, T(z)).
$$
\end{lemma}

For the proof of this lemma see Lemma~5.1 of \cite{ALR}.

\section{Finiteness of the number of eigenvalues of $H$}

In this section we prove the finiteness of the number of
eigenvalues of $H,$ that is, Theorem~\ref{THM 3}. We have divided
the proof into a sequence of lemmas.

\begin{lemma}\label{LEM 7}
There exist positive numbers $C_1, C_2, C_3$ and $\delta$ such
that the following inequalities hold\\
{\rm (i)} $C_1(x^2+y^2) \leq w_2(x,y)-m \leq
C_2(x^2+y^2),$ $x,y \in U_\delta(0);$\\
{\rm (ii)} $w_2(x,y)-m \geq C_3,$ $(x,y) \not\in U_\delta(0)
\times U_\delta(0).$
\end{lemma}

\begin{proof}
Since the function $w_2(\cdot,\cdot)$ is analytic on ${\Bbb
T}^2$ and it has a unique non-degenerate global minimum at the point
$(0,0) \in {\Bbb T}^2,$ the following decomposition holds
\begin{eqnarray*}
w_2(x,y)=m+\frac{1}{2} \left( \frac{\partial^2 w_2(0,0)}{\partial
x^2} x^2 + 2 \frac{\partial^2 w_2(0,0)}{\partial x \partial y} xy+
\frac{\partial^2 w_2(0,0)}{\partial y^2} y^2
\right)+O(|x|^3+|y|^3)
\end{eqnarray*}
as $x,y \to 0.$ Therefore, there exist positive numbers $C_1, C_2,
C_3$ and $\delta$ such that (i) and (ii) hold true.
\end{proof}

\begin{lemma}\label{LEM 8}
Let the assumption {\rm (iii)} of Theorem~$\ref{THM 3}$ be
fulfilled. Then there exists a positive number $C_1$ such that the
inequality $\Delta(x\,; z) \geq C_1$ holds for all $x \in {\Bbb
T}$ and $z \leq m.$
\end{lemma}

\begin{proof}
By assumption (iii) of Theorem~\ref{THM 3} we have $\min\limits_{x
\in {\Bbb T}} \Delta(x\,; m)>0.$ Since for any $x \in {\Bbb T}$
the function $\Delta(x\,; \cdot)$ is monotonically decreasing in
$(-\infty, m],$ we have
$$
\Delta(x\,; z) \geq \Delta(x\,; m) \geq \min\limits_{x \in {\Bbb
T}} \Delta(x\,; m)>0
$$
for all $x \in {\Bbb T}$ and $z \leq m.$ Now setting
$C_1:=\min\limits_{x \in {\Bbb T}} \Delta(x\,; m)$ we complete the
proof.
\end{proof}

We recall that by Lemma~\ref{LEM 1} the set $ \sigma$ is equal to
the set of all complex numbers $z \in {\Bbb C} \setminus [m_x,
M_x]$ such that $\Delta(x\,; z)=0$ for some $x \in {\Bbb T}.$

If the condition (i) or (ii) of Theorem~\ref{THM 3} holds, then by
the assertions (i) and (ii) of Theorem~\ref{THM 2} we have $E_{\rm
min} \in \sigma,$ hence there exists $x_1 \in {\Bbb T}$ such that
$\Delta(x_1\,; E_{\rm min})=0.$ Since $E_{\rm min}<m,$ the
function $\Delta(\cdot\,; E_{\rm min})$ is a regular in ${\Bbb
T}.$ Therefore, the number of zeros of this function is finite.

Let
$\{x \in {\Bbb T}: \Delta(x\,; E_{\rm min})=0\}=\{x_1,\ldots,x_n\}$
and $k_j$ be the multiplicity of $x_j$ for $j \in \{1,\ldots,n\}.$
The fact $E_{\rm min}<m$ implies that the difference $w_2(x,y)-z$
is positive for all $x,y \in {\Bbb T}$ and $z \le E_{\rm min}.$
Hence the function $(w_2(\cdot,\cdot)-z)^{-1}$ is an analytic one on
${\Bbb T}^2$ for all $z \le E_{\rm min}.$ Then there exists a
number $\delta>0$ such that for any $i,j \in \{1,\ldots,n\}$ and
$z \le E_{\rm min}$ the following representations are valid
\begin{eqnarray}
\frac{v_1(x)v_1(y)}{2(w_2(x,y)-z)}=\sum\limits_{k=0}^{[k_i/2]}
c_{ik}^{(1)}(z;x)(y-x_i)^k+(y-x_i)^{[k_i/2]+1}m_{i}^{(1)}(z;x,y),
\label{Eq5}
\end{eqnarray}
for all $x \in {\Bbb T}$ and $y \in U_\delta(x_i);$
\begin{eqnarray}
\frac{v_1(x)v_1(y)}{2(w_2(x,y)-z)}=\sum\limits_{k=0}^{[k_i/2]}
c_{ik}^{(2)}(z;y)(x-x_i)^k+(x-x_i)^{[k_i/2]+1}m_{i}^{(2)}(z;x,y),
\label{Eq6}
\end{eqnarray}
for all $x \in U_\delta(x_i)$ and $y \in {\Bbb T};$
\begin{eqnarray}
&&\frac{v_1(x)v_1(y)}{2(w_2(x,y)-z)}=\sum\limits_{k=0}^{[k_i/2]}
\sum\limits_{r=0}^{[k_j/2]} d_{ij}^{kr}(z)(x-x_i)^k(y-x_j)^r \nonumber\\
&&+ \sum\limits_{k=0}^{[k_i/2]} \sum\limits_{r=[k_j/2]+1}^\infty
d_{ij}^{kr}(z)(x-x_i)^k(y-x_j)^r\\
&&+\sum\limits_{k=[k_i/2]+1}^\infty \sum\limits_{r=0}^{[k_j/2]}
d_{ij}^{kr}(z)(x-x_i)^k(y-x_j)^r+ (x-x_i)^{[k_i/2]+1}
(y-x_j)^{[k_j/2]+1} q_{ij}(z;x,y),\nonumber
\label{Eq7}
\end{eqnarray}
for all $(x,y) \in U_\delta(x_i) \times U_\delta(x_j).$
Here $[\cdot]$ is the entire part of $a,$ for any $z \le E_{\rm min}$
the numbers $d_{ij}^{kr}(z)$ are some real coefficients, the
functions $c_{ik}^{(\alpha)}(z;\cdot),$ $\alpha=1,2;$
$m_{i}^{(1)}(z;\cdot,\cdot);$ $m_{i}^{(2)}(z;\cdot,\cdot)$ and
$q_{ij}^{kr}(z;\cdot,\cdot)$ are some analytic functions on ${\Bbb
T};$ ${\Bbb T} \times U_\delta(x_i);$ $U_\delta(x_i) \times {\Bbb
T}$ and $U_\delta(x_i) \times U_\delta(x_j),$ respectively.

\begin{lemma}\label{LEM 4}
Let the assumption {\rm (i)} or {\rm (ii)} of Theorem~$\ref{THM
3}$ be fulfilled and $j \in \{1,\ldots,n\}.$ Then there exist
numbers $C>0$ and $\delta>0$ such that the inequality
\begin{eqnarray}
\frac{|x-x_j|^{[k_j/2]+1}}{\sqrt{\Delta(x\,; z)}} \le C
\label{Eq4}
\end{eqnarray}
holds for all $x \in U_\delta(x_j)$ and $z \le E_{\rm min}.$
\end{lemma}

\begin{proof}
If the assumption {\rm (i)} or {\rm (ii)} of Theorem~\ref{THM 3}
holds, then by Theorem~\ref{THM 2} we have $E_{\rm min}<m.$ Since
the function $\Delta(x\,; \cdot)$ is monotonically decreasing on
$(-\infty, m)$ we have
$$
\frac{|x-x_j|^{[k_j/2]+1}}{\sqrt{\Delta(x\,; z)}} \le
\frac{|x-x_j|^{[k_j/2]+1}}{\sqrt{\Delta(x\,; E_{\rm min})}}
$$
for all $z \le E_{\rm min}.$ Taking into account the fact that the
number $k_j$ is the multiplicity of the $x_j$ and the function
$\Delta(\cdot\,; E_{\rm min})$ is analytic on ${\Bbb T}$ we
obtain the inequality \eqref{Eq4}.
\end{proof}

\begin{lemma}\label{LEM 5}
Let the assumptions of Theorem~$\ref{THM 3}$ be satisfied. Then
for any $z \le E_{\rm min}$ the operator $T(z)$ can be represented
in the form $T(z)=T_0(z)+T_1(z),$
where the operator-valued function $T_0(\cdot)$ is continuous in
the operator-norm in $(-\infty; E_{\rm min}]$ and $T_1(z)$ is a
finite-dimensional operator for all $z \le E_{\rm min}$ whose
dimension is independent of $z.$
\end{lemma}

\begin{proof}
Since the operators $T_{00}(z),$ $T_{01}(z)$ and $T_{01}^*(z)$ are of
rank one independently of $z,$ it is sufficient to study the
operator $T_{11}(z).$

We denote the kernel of the integral operator
$T_{11}(z)$ by $T_{11}(z;x,y),$ that is,
$$
T_{11}(z;x,y):=\frac{v_1(x)v_1(y)}{2\sqrt{\Delta(x\,;
z)}(w_2(x,y)-z) \sqrt{\Delta(y\,; z)}}.
$$

First we will prove the statement of lemma under the assumption (i) or
(ii) of Theorem~\ref{THM 3}. In this case $E_{\rm min}<m$ and
using the representations \eqref{Eq5}--\eqref{Eq7} we obtain
$
T_{11}(z)=T_{11}^0(z)+T_{11}^1(z),
$
where the kernels $T_{11}^0(z;x,y)$ and $T_{11}^1(z;x,y)$ of the
integral operators $T_{11}^0(z)$ and $T_{11}^1(z)$ has form
\begin{eqnarray*}
T_{11}^0(z;x,y):&&=(1-\chi_{V_\delta}(x))(1-\chi_{V_\delta}(y))
T_{11}(z;x,y)\\
&&+\frac{(1-\chi_{V_\delta}(x))}{\sqrt{\Delta(x\,; z)}}
\sum\limits_{i=1}^n \frac{\chi_{V_\delta}(y)
(y-x_i)^{[k_i/2]+1}}{\sqrt{\Delta(y\,; z)}}
M_i^{(1)}(z;x,y)\\
&&+\frac{(1-\chi_{V_\delta}(y))}{\sqrt{\Delta(y\,; z)}}
\sum\limits_{i=1}^n \frac{\chi_{V_\delta}(x)
(x-x_i)^{[k_i/2]+1}}{\sqrt{\Delta(x\,; z)}}
M_i^{(2)}(z;x,y)\\
&&+\chi_{V_\delta}(x) \chi_{V_\delta}(y) \sum\limits_{i,j=1}^n \frac{
(x-x_i)^{[k_i/2]+1}(y-x_j)^{[k_j/2]+1}}{\sqrt{\Delta(x\,; z)}
\sqrt{\Delta(y\,; z)}} Q_{ij}(z;x,y);
\end{eqnarray*}
\begin{eqnarray*}
T_{11}^1(z;x,y):&&= \frac{(1-\chi_{V_\delta}(x)) \chi_{V_\delta}(y)}
{\sqrt{\Delta(x\,; z)} \sqrt{\Delta(y\,; z)}} \sum\limits_{i=1}^n
\sum\limits_{k=0}^{[k_i/2]} (y-x_i)^k
c_{ik}^{(1)}(z;x)\\
&&+\frac{\chi_{V_\delta}(x) (1-\chi_{V_\delta}(y))}{\sqrt{\Delta(x\,; z)}
\sqrt{\Delta(y\,; z)}} \sum\limits_{i=1}^n
\sum\limits_{k=0}^{[k_i/2]} (x-x_i)^k
c_{ik}^{(2)}(z;y)\\
&&+\frac{\chi_{V_\delta}(x) \chi_{V_\delta}(y)}{\sqrt{\Delta(x\,; z)}
\sqrt{\Delta(y\,; z)}} \sum\limits_{i,j=1}^n \Bigl(
\sum\limits_{k=0}^{[k_i/2]} \sum\limits_{r=0}^{[k_j/2]}
d_{ij}^{kr}(z)(x-x_i)^k (y-x_j)^r\\
&&+ \sum\limits_{k=0}^{[k_i/2]} \sum\limits_{r=[k_j/2]+1}^\infty
d_{ij}^{kr}(z)(x-x_i)^k (y-x_j)^r
+ \sum\limits_{k=[k_i/2]+1}^\infty \sum\limits_{r=0}^{[k_j/2]}
d_{ij}^{kr}(z)(x-x_i)^k (y-x_j)^r \Bigr),
\end{eqnarray*}
respectively, where $V_\delta:=\bigcup\limits_{i=1}^n U_\delta(x_i),$ $\chi_A(\cdot)$ is the characteristic
function of the set $A \subset {\Bbb T},$
\begin{eqnarray*}
&&M_i^{(1)}(z;x,y):= \left \lbrace
\begin{array}{ll}
m_i^{(1)}(z;x,y), \,\, (x,y) \in {\Bbb T} \times U_\delta(x_i),\\
0, \qquad \qquad \quad \,\,\,\, {\rm otherwise},
\end{array} \right.\\
&&M_i^{(2)}(z;x,y):= \left \lbrace
\begin{array}{ll}
m_i^{(2)}(z;x,y), \,\, (x,y) \in U_\delta(x_i) \times {\Bbb T},\\
0, \qquad \qquad \quad \,\,\,\, {\rm otherwise},
\end{array} \right.\\
&&Q_{ij}(z;x,y):= \left \lbrace
\begin{array}{ll}
q_{ij}(z;x,y), \,\, (x,y) \in U_\delta(x_i) \times U_\delta(x_j),\\
0, \qquad \qquad \quad \,\,\,\, {\rm otherwise}.
\end{array} \right.
\end{eqnarray*}

Applying Lemma~\ref{LEM 4} we obtain that the function
$T_{11}^0(z;\cdot,\cdot),$ $z \le E_{\rm min}$ is
square-integrable on ${\Bbb T}^2$ and converges almost everywhere
to $T_{11}^0(E_{\rm min};\cdot,\cdot)$ as $z \to E_{\rm min}-0.$
Then by the Lebesgue dominated convergence theorem the operator
$T_{11}^0(z)$ converges in the operator-norm to $T_{11}^0(E_{\rm min})$ as
$z \to E_{\rm min}-0.$ The finite dimensionality of the operator
$T_{11}^1(z)$ follows from the definition of $T_{11}^1(z;x,y).$
Now setting
$$
T_0(z):=\left( \begin{array}{cc}
0 & 0\\
0 & T_{11}^0(z)\\
\end{array}
\right), \quad T_1(z):=\left( \begin{array}{cc}
T_{00}(z) & T_{01}(z)\\
T_{01}^*(z) & T_{11}^1(z)\\
\end{array}
\right)
$$
we complete proof of Lemma~\ref{LEM 5} under the assumption (i) or
(ii) of Theorem~\ref{THM 3}.

Let the assumption (iii) of Theorem~\ref{THM 3} be satisfied. Then
by Theorem~\ref{THM 2} we have $E_{\rm min}=m.$ Applying Lemmas~\ref{LEM 7}
and \ref{LEM 8} and as well as inequality \eqref{Eq1} one can see
that the function $|T_{11}(z;\cdot,\cdot)|$ can be estimated by
$$
C_1 \left(1+\frac{|x|^\alpha |y|^\alpha}{x^2+y^2} \right)
$$
for $z \le m$ with $\alpha \ge 1.$ The latter function is a
square-integrable on ${\Bbb T}^2$ and the function
$T_{11}(z;\cdot,\cdot)$ converges almost everywhere to
$T_{11}(m;\cdot,\cdot)$ as $z \to m-0.$ Then by the Lebesgue dominated
convergence theorem the operator $T_{11}(z)$ converges in the norm
to $T_{11}(m)$ as $z \to m-0.$ Now setting
$$
T_0(z):=\left( \begin{array}{cc}
0 & 0\\
0 & T_{11}(z)\\
\end{array}
\right), \quad T_1(z):=\left( \begin{array}{cc}
T_{00}(z) & T_{01}(z)\\
T_{01}^*(z) & 0\\
\end{array}
\right)
$$
we complete proof of Lemma~\ref{LEM 5} under the assumption (iii) of
Theorem~\ref{THM 3}.
\end{proof}

We are now ready for the proof of Theorem~$\ref{THM 3}.$

\begin{proof2}
Using the Weyl inequality
\begin{eqnarray}
n(\lambda_1+\lambda_2, A_1+A_2) \leq n(\lambda_1,
A_1)+n(\lambda_2, A_2)
\label{Weyl ineq}
\end{eqnarray}
for the sum of compact operators $A_1$ and $A_2$ and for any
positive numbers $\lambda_1$ and $\lambda_2$ we have
\begin{eqnarray}
n(1, T(z)) &&\leq n(2/3, T_0(z))+n(1/3, T_1(z))\nonumber\\
&&\leq n(1/3, T_0(z)-T_0(E_{\min}))+n(1/3, T_0(E_{\min}))+n(1/3,
T_1(z))
\label{Eq8}
\end{eqnarray}
for all $z<E_{\min}.$

By virtue of Lemma~\ref{LEM 5} the operator $T_0(E_{\rm min})$ is
compact and hence $n(1/3, T_0(E_{\min}))<\infty$ and $n(1/3,
T_0(z)-T_0(E_{\min}))$ tends to zero as $z \to E_{\min}-0.$ Since
$T_1(z)$ is a finite-dimensional operator and its dimension is
independent of $z,$ $z<E_{\min},$ there exists a number $F$ such
that for all $z<E_{\min},$ we have $n(1/3, T_1(z)) \le F<\infty.$
So, by the inequality \eqref{Eq8} we obtain that the number $n(1,
T(z))$ is finite for all $z<E_{\min}.$

Now Lemma~\ref{LEM 3} implies that $N(z)=n(1, T(z))$ as
$z<E_{\min}$ and hence
$$
\lim\limits_{z\to E_{\min}-0}N(z)= N(E_{\min}) \leq n(1/3,
T_0(E_{\min}))+n(1/3, T_1(E_{\min}))<\infty.
$$
It means that the number of eigenvalues of $H$ lying on the
left of $E_{\min}$ is finite.
\end{proof2}

\section{Infiniteness of the number of eigenvalues of $H$}

In this section we consider the case when the parameter functions
$v_i(\cdot),$ $i=1,2,$ $w_1(\cdot)$ and $w_2(\cdot,\cdot)$ have the
special forms:
\begin{eqnarray*}
&&v_0(x):=0, \quad w_1(x):=a, \quad v_1(x):=b, \quad a, b \in {\Bbb
R} \setminus \{0\};\\
&&w_2(x,y):=\varepsilon(x-y), \quad \varepsilon(x):=1-\cos x.
\end{eqnarray*}

It is obvious that the function $w_2(\cdot,\cdot)$ has
non-degenerate minimum at the points of the form $(x,x)$ for any
$x \in {\Bbb T}.$
Then it is clear that the number $z=w_0$ is an eigenvalue of $H$
with the associated eigenvector $f=(f_0,0,0)$ with $f_0 \neq 0$
and the equality holds
$
\sigma_{\rm ess}(H)=\{E_{\rm min}\} \cup [0, 2] \cup \{E_{\rm
max}\},
$
where $E_{\rm min}$ and $E_{\rm max}$ are zeros of the function
$\Delta(\cdot)$ defined on ${\Bbb C} \setminus [0, 2]$ by
$$
\Delta(z):=a-z-\frac{b^2}{2}\int_{\Bbb T}
\frac{dt}{\varepsilon(t)-z}
$$
such that $E_{\rm min}<0$ and $E_{\rm max}>2.$

We define the function $D(\cdot)$ on ${\Bbb C} \setminus
\sigma_{\rm ess}(H)$ as
$$
D(z):=\prod_{k=0}^\infty D_k(z),\quad
D_k(z):=1-\frac{1}{2 \Delta(z)} d_k(z), \quad
d_k(z):=\int_{\Bbb T}
\frac{\cos(kt)dt}{\varepsilon(t)-z}.
$$

The following lemma establishes a connection between eigenvalues
of the operator $H$ and zeros of the function $D(\cdot).$

\begin{lemma}\label{LEM 9}
The number $z \in {\Bbb C} \setminus (\sigma_{\rm ess}(H) \cup
\{w_0\})$ is an eigenvalue of $H$ if and only if $D(z)=0.$
Moreover, if for some $k \in {\Bbb N}$ the number $z_k \in {\Bbb
C} \setminus \sigma_{\rm ess}(H)$ is an eigenvalue of $H$ with
$D_k(z_k)=1-\lambda_k(z_k)=0,$ then the corresponding eigenvector
$f^{(k)}$ has the form $f^{(k)}:=(0,f_1^{(k)},f_2^{(k)}),$ where the
functions $f_1^{(k)}$ and $f_2^{(k)}$ are defined by
\begin{eqnarray}
f_1^{(k)}(x):=\exp(\pm ikx), \quad f_2^{(k)}(x,y):=\frac{b
(f_1^{(k)}(x)+f_1^{(k)}(y))}{2(\varepsilon(x-y)-z_k)}.
\label{Eq12}
\end{eqnarray}
\end{lemma}

\begin{proof}
Let the number $z \in {\Bbb C} \setminus (\sigma_{\rm ess}(H) \cup
\{w_0\})$ be an eigenvalue of $H$ and $f=(f_0,f_1,f_2) \in
{\mathcal H}$ be the corresponding eigenvector. Then $f_0,$ $f_1$
and $f_2$ satisfy the following system of equations
\begin{eqnarray}
&&(w_0-z)f_0=0;\nonumber\\
&&(a-z)f_1(x)+b \int_{\Bbb T} f_2(x,t)dt=0;\\
&&\frac{b}{2}(f_1(x)+f_1(y))+(\varepsilon(x-y)-z)f_2(x,y)=0.\nonumber
\label{Eq9}
\end{eqnarray}

Using the condition $z \neq w_0$ we get from the first equation of
\eqref{Eq9} that $f_0=0.$ Since $z \not\in [0, 2],$ from the
third equation of the system \eqref{Eq9} for $f_2$ we find
\begin{eqnarray}
f_2(x,y)=-\frac{b(f_1(x)+f_1(y))}{2(\varepsilon(x-y)-z)}.
\label{Eq10}
\end{eqnarray}

Substituting the expression \eqref{Eq10} for $f_2$ into the second
equation of the system \eqref{Eq9} and using the fact that
$\Delta(z) \neq 0$ for any $z \in {\Bbb C} \setminus \sigma_{\rm
ess}(H)$ we conclude that the number $z \in {\Bbb C} \setminus
(\sigma_{\rm ess}(H) \cup \{w_0\})$ is an eigenvalue of $H$ if and
only if the number 1 is an eigenvalue of the integral operator
$\widetilde{T}(z)$ in $L_2({\Bbb T})$ with the kernel
$$
\frac{b^2}{2 \Delta(z) (\varepsilon(x-y)-z)}.
$$

Since the function $(\varepsilon(\cdot)-z)^{-1}$ is continuous
on ${\Bbb T}$ and $\Delta(z) \neq 0$ for all $z \in {\Bbb C}
\setminus \sigma_{\rm ess}(H),$ the operator $\widetilde{T}(z)$ is
Hilbert-Schmidt and as well trace
class. Hence, the determinant ${\rm det}(I-\widetilde{T}(z))$ of
the operator $I-\widetilde{T}(z)$ exists and is given by the
formula (see Theorem~XIII.106 of \cite{RS4})
\begin{eqnarray}
{\rm det}(I-\widetilde{T}(z))=\prod\limits_{k=0}^\infty
(1-\lambda_k(z)),
\label{Eq11}
\end{eqnarray}
where $I$ is the identity operator on $L_2({\Bbb T})$ and the
numbers $\{\lambda_k(z)\}$ are the eigenvalues of
$\widetilde{T}(z)$ counted with their algebraic multiplicities.
By Theorem~XIII.105 of \cite{RS4} the number 1 is an eigenvalue of
$\widetilde{T}(z)$ if and only if ${\rm
det}(I-\widetilde{T}(z))=0.$

Let $\varphi$ be the eigenfunction of $\widetilde{T}(z)$
associated with the eigenvalue $\lambda,$ that is,
$$
\lambda \varphi(x)=\frac{b^2}{2 \Delta(z)} \int_{\Bbb T}
\frac{\varphi(t)dt}{\varepsilon(x-t)-z}.
$$
By expanding $\varphi$ into a series with respect to the basis $\{\exp(ikx)\}_{k
\in {\Bbb Z}}$ we obtain
$$
\lambda c_k \exp(ikx)=\frac{b^2 c_k}{2 \Delta(z)}
\int_{\Bbb T} \frac{\exp(ikt) dt}{\varepsilon(x-t)-z}
$$
or $\lambda(z)=b^2 d_k(z)/(2 \Delta(z)).$ Then  for any $k \in
{\Bbb Z}$ the eigenvalue $\lambda_k(z)$ of the operator
$\widetilde{T}(z)$ in formula \eqref{Eq11} can be expressed by $
\lambda_k(z)=b^2 d_k(z)/(2 \Delta(z))$ and the corresponding
eigenfunction $\varphi_k(\cdot)$ has the form
$\varphi_k(x):=\exp(ikx).$

For any $k \in {\Bbb Z}$ and $z \in {\Bbb C} \setminus \sigma_{\rm
ess}(H)$ we have $d_k(z)=d_{-k}(z).$ Hence, if the number
$1=\lambda_k(z_k),$ $z_k \in {\Bbb C} \setminus \sigma_{\rm
ess}(H)$ is an eigenvalue of $\widetilde{T}(z_k),$ then
$\varphi_k(x)=\exp(\pm ikx)$ is the corresponding eigenfunction of
$\widetilde{T}(z_k).$ From here it follows that if the number $z_k
\in {\Bbb C} \setminus \sigma_{\rm ess}(H)$ is an eigenvalue of
$H$ with $D_k(z_k)=1-\lambda_k(z_k)=0,$ then the corresponding
eigenvector $f^{(k)}$ has the form $f^{(k)}:=(0,f_1^{(k)},f_2^{(k)}),$
where $f_1^{(k)}$ and $f_2^{(k)}$ are defined by \eqref{Eq12}.
\end{proof}

Let ${\Bbb N}_0:={\Bbb N} \cup \{0\}.$

\begin{lemma}\label{LEM 10}
For the functions $\Delta(\cdot)$ and $d_k(\cdot),$ $k \in {\Bbb
N}_0,$ the equalities hold
\begin{eqnarray*}
&&\Delta(z)=a-z-\frac{\pi b^2}{\sqrt{z^2-2z}},\quad d_k(z)=\frac{2 \pi \left[ 1-z-\sqrt{z^2-2z}
\right]^k}{\sqrt{z^2-2z}}, \quad z<0;\\
&&\Delta(z)=a-z+\frac{\pi b^2}{\sqrt{z^2-2z}},\quad d_k(z)=\frac{2 \pi \left[ 1-z+\sqrt{z^2-2z}
\right]^k}{\sqrt{z^2-2z}}, \quad z>2.
\end{eqnarray*}
\end{lemma}

\begin{proof}
The assertion of lemma for the case $z<0$ can be proven similarly
to Lemma~10 of \cite{MA}. We consider the case $z>2.$ Using the identity
$$
\int_0^\pi \frac{\cos(kt)dt}{1+2c \cos t+c^2}=\frac{\pi
(-c)^k}{1-c^2},
$$
where $k \in {\Bbb N}_0$ and $0<c<1,$ one can show that
\begin{eqnarray}
\int_0^\pi \frac{\cos(kt)dt}{1+(2c/(1+c^2)) \cos t}
=\frac{\pi (-c)^k}{\sqrt{1-(2c)^2/(1+c^2)^2}}.
\label{Eq13}
\end{eqnarray}

Since the function $\varepsilon(\cdot)$ is an even, the function
$d_k(\cdot)$ has form
$$
d_k(z)=2\int_0^\pi \frac{\cos(kt)dt}{1-\cos t-z}=
\frac{2}{1-z} \int_0^\pi \frac{\cos(kt)dt}{1+\cos t/(z-1)}.
$$

Introducing the notation $c_z:=z-1-\sqrt{z^2-2z}$ we obtain
$$
d_k(z)=\frac{2}{1-z} \int_0^\pi \frac{\cos(kt)dt}{1+2c_z
\cos t/(1+c_z^2)}.
$$

It is clear that $c_z \in (0, 1)$ for all $z>2$ and hence the
equality \eqref{Eq13} completes proof of lemma for the case $z>2.$
\end{proof}

Now we formulate the result about infiniteness of the discrete
spectrum of $H.$

\begin{theorem}\label{THM 4}
{\rm (i)} The operator $H$ has an infinite number of eigenvalues
$\{\xi_k^{(\alpha)}\}_0^\infty$ with $\alpha=1,2,3$ such that
$\{\xi_k^{(1)}\}_0^\infty \subset (-\infty, E_{\rm min}),$
$\{\xi_k^{(2)}\}_0^\infty \subset (E_{\max}, \infty),$
$\{\xi_k^{(3)}\}_0^\infty \subset (2, E_{\max})$
and
$$
\lim\limits_{k \to \infty} \xi_k^{(1)}=E_{\rm min}, \quad \lim\limits_{k \to \infty}
\xi_k^{(2)}=\lim\limits_{k \to \infty} \xi_k^{(3)}=E_{\rm max}.
$$
For $\alpha=1,2,3$ the multiplicity of every eigenvalue
$\xi_k^{(\alpha)},$ $k \in {\Bbb N}$ is two, the multiplicity of
$\xi_0^{(3)}$ is one or two and $\xi_0^{(1)},$ $\xi_0^{(2)}$ are
simple eigenvalues of $H.$ Moreover, the eigenvalues $\xi_k^{(1)}$
resp. $\xi_k^{(2)}$ are solutions of the rational equations
$$
\frac{\pi
b^2}{\Delta(z)}\frac{\left[1-z-\sqrt{z^2-2z}\right]^k}{\sqrt{z^2-2z}}=1,
\quad z<E_{\rm min};
$$
resp.
$$
\frac{\pi b^2}{\Delta(z)}
\frac{\left[1-z+\sqrt{z^2-2z}\right]^k}{\sqrt{z^2-2z}}=1, \quad
z>E_{\rm max}.
$$
{\rm (ii)} The operator $H$ has no eigenvalues in $(E_{\min}, 0).$
\end{theorem}

\begin{proof}
(i) For any fixed $k \in {\Bbb N}_0$ we have
$$
\lim\limits_{z \to \pm \infty} D_k(z)=1, \quad \lim\limits_{z \to
E_{\rm min}-0} D_k(z)=\lim\limits_{z \to E_{\rm max}+0}
D_k(z)=-\infty.
$$

Since the function $D_k(\cdot)$ is continuous in $(-\infty,
E_{\rm min})$ and $(E_{\max}, \infty)$ there exist numbers
$\xi_k^{(1)} \in (-\infty, E_{\rm min})$ and $\xi_k^{(2)} \in
(E_{\rm max}, \infty)$ such that $D_k(\xi_k^{(\alpha)})=0$ for
$\alpha=1,2.$ The equality $\lim\limits_{z \to \pm \infty} D(z)=1$
and the analyticity of the function $D(\cdot)$ on ${\Bbb C}
\setminus (\{E_{\rm min}\} \cup [0, 2] \cup \{E_{\rm max}\})$
imply that $\lim\limits_{k \to \infty} \xi_k^{(1)}=E_{\rm min}$
and $\lim\limits_{k \to \infty} \xi_k^{(2)}=E_{\rm max}.$
By Lemma~\ref{LEM 9} for any $\alpha=1,2$ and $k \in {\Bbb N}_0$ the number
$\xi_k^{(\alpha)}$ is an eigenvalue of $H$ and the corresponding
eigenvector $f^{(k)}$ has the form $f^{(k)}:=(0,f^{(k)}_1,f^{(k)}_2),$
where $f^{(k)}_1$ and $f^{(k)}_2$ are defined by \eqref{Eq12} with
$z_k=\xi_k^{(\alpha)}.$ Moreover, $\xi_0^{(1)},$ $\xi_0^{(2)}$ are simple
eigenvalues and for any $k \in {\Bbb N}$ the multiplicities of
$\xi_k^{(\alpha)}$ are two.

Note that the function $D_k(\cdot)$ is defined on $(2, E_{\rm max})$ and for $z>2$ we have
\begin{eqnarray*}
-1<1-z+\sqrt{z^2-2z}<0.
\end{eqnarray*}
By Lemma~\ref{LEM 10} the function $D_k(\cdot)$ can be rewritten
as
\begin{eqnarray*}
D_k(z)=1-\frac{1}{2 \Delta(z)}
\frac{\left[1-z+\sqrt{z^2-2z}\right]^k}{\sqrt{z^2-2z}}, \quad z
\in (2, E_{\rm max});
\end{eqnarray*}
therefore, for any fixed $z \in (2, E_{\rm max})$ the equality $\lim\limits_{k \to \infty} D_k(z)=1$ holds.

It is clear that $\Delta(z)>0$ for all $z \in (2, E_{\rm max})$ and hence the inequality $D_{2k+1}(z)>1$ holds
for all $k \in {\Bbb N}_0$ and $z \in (2, E_{\rm max}).$
Since $\lim\limits_{k \to \infty} D_k(z)=1$ for any fixed $z \in (2, E_{\rm max}),$
there exists a subsequence $\{k_n\} \subset 2{\Bbb N}_0$ such that
$D_{k_n}((E_{\rm max}+2)/2)>0$ holds for any $n \in {\Bbb N}_0.$
Now the equality $\lim\limits_{z \to E_{\rm max}-0} D_{2k}(z)=-\infty$ and the continuity
of the function $D_k(\cdot)$ imply that
$D_{k_n}(\xi_n^{(3)})=0$ for some
$\xi_n^{(3)} \in ((E_{\rm max}+2)/2, E_{\rm max}).$ It follows
from the analyticity of $D(\cdot)$ on ${\Bbb C} \setminus (\{E_{\rm
min}\} \cup [0, 2] \cup \{E_{\rm max}\})$ that
$\lim\limits_{n \to \infty} \xi_n^{(3)}=E_{\rm max}.$ Now repeated
application of Lemma~\ref{LEM 9} implies that the number $\xi_n^{(3)}$ is an
eigenvalue of $H.$ Similarly, for any $n \in {\Bbb N}$ the multiplicity of
$\xi_n^{(3)}$ is two. If $k_0=0,$ then $\xi_0^{(3)}$ is a simple, otherwise
its multiplicity is two.

(ii) It is clear that $0<1-z-\sqrt{z^2-2z}<1$ for $z<0.$
Then by Lemma~\ref{LEM 10} the function $D_k(\cdot)$ can be represented
as
\begin{eqnarray*}
D_k(z)=1-\frac{1}{2 \Delta(z)}
\frac{\left[1-z-\sqrt{z^2-2z}\right]^k}{\sqrt{z^2-2z}}, \quad z
\in (E_{\rm min}; 0).
\end{eqnarray*}
Since $\Delta(z)<0$ for all $z \in (E_{\rm min}, 0)$ the inequality $D_k(z)>1$ holds for such $z$
and hence $D(z)>1.$ By Lemma~\ref{LEM 9} the operator has no eigenvalues in $(E_{\rm min}, 0).$
\end{proof}

\section{The case $v_1(0)=0$ and $\Delta(0\,; m)=0$}

In this section we are going to discuss the discrete spectrum of $H$ for the case $v_1(0)=0$
and $\Delta(0\,; m)=0.$ In this case the discrete spectrum of $H$ might be finite or infinite
depending on the behavior of the parameter functions.

{\bf Case I: Infiniteness.} Let the parameter functions $v_1(\cdot),$ $w_1(\cdot)$ and $w_2(\cdot,\cdot)$
have the form
\begin{eqnarray}
&&v_1(x)=\sqrt{\mu} \sin x,\quad \mu>0;\quad
 w_1(x)=1+\sin^2 x;\nonumber \\
&& w_2(x,y)=\varepsilon(x)+l \varepsilon(x+y)+\varepsilon(y), \quad \varepsilon(x):=1-\cos x,\,\, l>0.
 \label{parametrs H}
\end{eqnarray}

Then the function $w_2(\cdot,\cdot)$ has a unique non-degenerate zero minimum ($m=0$) at the point $(0,0) \in {\Bbb T}^2$
and $v_1(0)=0.$ It is easy to see that for
$$
\Delta(x\,; z)=1+\sin^2 x-z-\frac{\mu}{2} \int_{\Bbb T} \frac{\sin^2 t\, dt}
{\varepsilon(x)+l \varepsilon(x+t)+\varepsilon(t)-z}
$$
we have $\Delta(0\,; 0)=0$ if and only if
$$
\mu=\mu_0:=(1+l)\left(\int_0^\pi \frac{\sin^2 t\, dt}{\varepsilon(t)}\right)^{-1}=\frac{1+l}{\pi}.
$$

The following decomposition plays an important role in the proof of the infiniteness of the discrete spectrum of $H.$

\begin{lemma}\label{decomp of Delta}
The following decomposition
$$
\Delta(x\,; z)=\Delta(0\,; 0)+ \frac{\mu\pi(1+2l-l^2)}{(1+l)^2
\sqrt{1+2l}} \sqrt{x^2-\frac{2(1+l)}{1+2l}z}+O(x^2)+O(\sqrt{|z|})
$$
holds as $x \to 0$ and $z \to -0.$
\end{lemma}

\begin{proof}
Let $\delta>0$ be sufficiently small and ${\Bbb T}_\delta:={\Bbb T} \setminus (-\delta, \delta).$
We rewrite the function $\Delta(\cdot\,; \cdot)$
in the form
$
\Delta(x\,; z)=\Delta_1(x\,; z)+\Delta_2(x\,; z),
$
where
\begin{eqnarray*}
 \Delta_1(x\,; z):&&=1+\sin^2 x-z-\frac{\mu}{2} \int_{{\Bbb T}_\delta} \frac{\sin^2 t\, dt}
 {\varepsilon(x)+l \varepsilon(x+t)+\varepsilon(t)-z},\\
 \Delta_2(x\,; z):&&=-\frac{\mu}{2} \int_{-\delta}^\delta \frac{\sin^2 t\, dt}
 {\varepsilon(x)+l \varepsilon(x+t)+\varepsilon(t)-z}.
\end{eqnarray*}
Since $\Delta_1(\cdot\,; z)$ is an even analytic function on ${\Bbb T}$ for any $z \le 0,$ we have
\begin{eqnarray}
\Delta_1(x\,; z)=\Delta_1(0\,; 0)+O(x^2)+O(|z|)
\label{decomp for Delta1}
\end{eqnarray}
as $x \to 0$ and $z \to -0.$
Using
\begin{eqnarray}
\sin x=x+O(x^3), \quad 1-\cos x=\frac{1}{2}x^2+O(x^4), \quad x\to 0
\label{decomp for sin and cos}
\end{eqnarray}
we obtain
$$
\Delta_2(x\,; z):=-\mu \int_{-\delta}^\delta \frac{t^2 dt}{(1+l)x^2+2lxt+(1+l)t^2-2z}+O(x^2)+O(|z|)
$$
as $x \to 0$ and $z \to -0.$
For the convenience we rewrite the latter integral as
\begin{eqnarray*}
&&\int_{-\delta}^\delta \frac{t^2 dt}{(1+l)x^2+2lxt+(1+l)t^2-2z}\\
&&=\frac{2\delta}{1+l}
-\frac{lx}{1+l} \int_{-\delta}^\delta \frac{2t dt}{(1+l)x^2+2lxt+(1+l)t^2-2z}\\
&&-\frac{(1+l)x^2-2z}{1+l} \int_{-\delta}^\delta \frac{dt}{(1+l)x^2+2lxt+(1+l)t^2-2z}.
\end{eqnarray*}
Now we study each integral in the last equality. For the integral in the second summand
we obtain
\begin{eqnarray*}
\int_{-\delta}^\delta \frac{2t dt}{(1+l)x^2+2lxt+(1+l)t^2-2z}
&&=\frac{1}{1+l} \log \left| 1+\frac{4lx\delta}
{(1+l)x^2-2lx\delta+(1+l)\delta^2-2z} \right| \\
&&-\frac{2lx}{1+l} \int_{-\delta}^\delta \frac{dt}{(1+l)x^2+2lxt+(1+l)t^2-2z}.
\end{eqnarray*}
Since
$$
\log \left| 1+\frac{4lx\delta}
{(1+l)x^2-2lx\delta+(1+l)\delta^2-2z} \right|=O(x)
$$
as $x \to 0,$ comparing the last expressions we obtain
\begin{eqnarray*}
&&\int_{-\delta}^\delta \frac{t^2 dt}{(1+l)x^2+2lxt+(1+l)t^2-2z}=\frac{2\delta}{1+l}\\
&&-\left(\frac{1+2l-l^2}{(1+l)^2}x^2-\frac{2}{1+l}z \right)
\int_{-\delta}^\delta \frac{dt}{(1+l)x^2+2lxt+(1+l)t^2-2z}+O(x^2)+O(|z|)
\end{eqnarray*}
as $x \to 0$ and $z \to -0.$ Using the identity
\begin{eqnarray}
\int_a^b \frac{dt}{x^2+t^2}=\frac{1}{|x|}\left(\arctan \frac{b}{|x|}-\arctan \frac{a}{|x|}\right)
\label{arctan 1}
\end{eqnarray}
we have
\begin{eqnarray*}
&&\int_{-\delta}^\delta \frac{dt}{(1+l)x^2+2lxt+(1+l)t^2-2z}=\frac{1}{1+l} \int_{-\delta}^\delta \frac{dt}{(t+\frac{l}{1+l}x)^2+\frac{1+2l}{(1+l)^2}x^2-\frac{2}{l+1}z}\\
&&=\frac{1}{(1+l) \sqrt{\frac{1+2l}{(1+l)^2}x^2-\frac{2}{l+1}z}} \Bigl(\arctan \frac{\delta+\frac{l}{1+l}x}{\sqrt{ \frac{1+2l}{(1+l)^2}x^2-\frac{2}{l+1}z }}+\arctan \frac{\delta-\frac{l}{1+l}x}{\sqrt{\frac{1+2l}{(1+l)^2}x^2-\frac{2}{l+1}z}} \Bigr).
\end{eqnarray*}

The following properties of the $\arctan$ function
\begin{eqnarray}
\arctan y+\arctan \frac{1}{y}=\frac{\pi}{2}, \quad y \geq 0
\quad \mbox{and} \quad
\arctan y=O(y), \quad y \to 0
\label{arctan 2}
\end{eqnarray}
imply that
\begin{eqnarray*}
&&\int_{-\delta}^\delta
\frac{dt}{(1+l)x^2+2lxt+(1+l)t^2-2z}=\\
&&\left(\frac{1+2l-l^2}{(1+l)^2}x^2-\frac{2}{1+l}z
\right) \frac{\pi}{(1+l) \sqrt{\frac{1+2l}{(1+l)^2}x^2-\frac{2}{l+1}z}}
+O(\sqrt{\frac{1+2l}{(1+l)^2}x^2-\frac{2}{l+1}z})
\end{eqnarray*}
as $x\to 0$ and $z \to -0.$ Taking into account
\begin{eqnarray*}
&&\left(\frac{1+2l-l^2}{(1+l)^2}x^2-\frac{2}{1+l}z
\right) \frac{\pi}{(1+l) \sqrt{
\frac{1+2l}{(1+l)^2}x^2-\frac{2}{l+1}z}}\\
&&=\pi\frac{1+2l-l^2}{(1+l)^2\sqrt{1+2l}}  \sqrt{
x^2-\frac{2(1+l)}{2l+1}z }+O(\sqrt{-z}),
\end{eqnarray*}
we obtain
\begin{eqnarray}
  \Delta_2(x\,;z)=\Delta_2(0\,;0)+\frac{\mu\pi(1+2l-l^2)}{(1+l)^2 \sqrt{1+2l}}
  \sqrt{x^2-\frac{2(1+l)}{2l+1}z }+O(x^2) +O(\sqrt{-z})
  \label{decomp for Delta2}
\end{eqnarray}
as $x \to 0.$ The equalities (\ref{decomp for Delta1}) and (\ref{decomp for Delta2})
give the proof of lemma.
\end{proof}

Let $T(\delta; z)$ be the operator in ${\cal H}_0 \oplus {\cal H}_1$
defined by
$$
T(\delta; z):=\left( \begin{array}{cc}
0 & 0\\
0 & T_{11}(\delta; z)\\
\end{array}
\right),
$$
where $T_{11}(\delta; z)$ is the integral operator on $L_2({\Bbb
T})$ with the kernel
$$
\frac{1}{\pi} \frac{(1+l)^2 \sqrt{1+2l}}{1+2l-l^2}
\frac{1}{\sqrt[4]{x^2-\frac{2(1+l)}{2l+1}z }}
\frac{\chi_{(-\delta;\delta)}(x)\chi_{(-\delta;\delta)}(y)xy}{(1+l)x^2+2lxy+(1+l)y^2-2z}
\frac{1}{\sqrt[4]{ y^2-\frac{2(1+l)}{2l+1}z}}.
$$

\begin{lemma}\label{T(0) compact}
Let $\mu=\mu_0.$ Then for any $z \le 0$ the operator
$F(z):=T(z)-T(\delta; z)$ is compact and the operator-valued
function $F(\cdot)$ is continuous in the operator-norm in $(-\infty, 0].$
\end{lemma}

\begin{proof}
Denote by $T_{11}(z;x,y)$ and $T_{11}(\delta,z;x,y)$ the kernel of the operator
$T_{11}(z)$ and $T_{11}(\delta;z),$ respectively, and set $F(z;x,y):=T_{11}(z;x,y)-T_{11}(\delta,z;x,y).$
We split the function $F(z;\cdot,\cdot),$ $z<0$ into four parts
$$
F(z;x,y)=F_0(z;x,y)+F_1(z;x,y)+F_2(z;x,y)+F_3(z;x,y),
$$
where
\begin{eqnarray*}
F_0(z;x,y):&&=(1-\chi_{(-\delta, \delta)}(x)\chi_{(-\delta, \delta)}(y))T_{11}(z;x,y),\\
F_1(z;x,y):&&=\frac{\mu}{2} \frac{\chi_{(-\delta, \delta)}(x)}{\sqrt{\Delta(x\,; z)}}\frac{\chi_{(-\delta, \delta)}(y)}{\sqrt{\Delta(y\,; z)}}\\
&&\times \Bigg( \frac{\sin x \sin
y}{\varepsilon(x)+l\varepsilon(x+y)+\varepsilon(y)-z}- \frac{2xy}
{(1+l)x^2+2lxy+(1+l)y^2-2z}\Bigg),\\
F_2(z;x,y):&&=\Bigg(\frac{\mu}{\sqrt{\Delta(x\,;
z)}}-\frac{\mu}{\sqrt{\frac{\mu\pi(1+2l-l^2)}{(1+l)^2 \sqrt{1+2l}}
\sqrt{x^2-\frac{2(1+l)}{1+2l}z}}} \Bigg)
\frac{1}{\sqrt{\Delta(y\,; z)}}\\
&&\times \frac{\chi_{(-\delta, \delta)}(x)\chi_{(-\delta, \delta)}(y)xy}{(1+l)x^2+2lxy+(1+l)y^2-2z},\\
F_3(z;x,y):&&=
\frac{1}{\sqrt{\frac{\mu\pi(1+2l-l^2)}{(1+l)^2 \sqrt{1+2l}}
\sqrt{x^2-\frac{2(1+l)}{1+2l}z}}}\frac{\chi_{(-\delta, \delta)}(x)\chi_{(-\delta, \delta)}(y)xy}
{(1+l)x^2+2lxy+(1+l)y^2-2z}\\
&&\times \Bigg(\frac{\mu}{\sqrt{\Delta(y\,; z)}}
-\frac{\mu}{\sqrt{\frac{\mu\pi(1+2l-l^2)}{(1+l)^2 \sqrt{1+2l}}
\sqrt{y^2-\frac{2(1+l)}{1+2l}z}}} \Bigg).
\end{eqnarray*}

We show that the functions $F_i(z;\cdot,\cdot),$
$i=0,1,2,3$ are square-integrable on ${\Bbb T}^2$ for any
fixed $z \leq 0.$
First we note that for any fixed $z \leq 0$ the function $F_0(z;\cdot,\cdot)$ is bounded on
${\Bbb T}^2$ and hence it is a square-integrable on this set.

Using the decompositions (\ref{decomp for sin and cos}) we obtain that there exists $C>0$
such that for any $z \leq 0$ the inequality
$$
\Bigg| \frac{\sin x \sin
y}{\varepsilon(x)+l\varepsilon(x+y)+\varepsilon(y)-z}- \frac{2xy}
{(1+l)x^2+2lxy+(1+l)y^2-2z}\Bigg| \leq C |xy|, \quad x,y \in (-\delta, \delta)
$$
holds. Therefore, for any fixed $z \leq 0$ the function $F_1(z;\cdot,\cdot)$ is a square-integrable on ${\Bbb T}^2.$

By Lemma~\ref{decomp of Delta} for any $x \in (-\delta, \delta)$ and $z \in (-\delta, 0)$ we get the
estimate
$$
\Bigg |\frac{\mu}{\sqrt{\Delta(x\,;
z)}}-\frac{\mu}{\sqrt{\frac{\mu\pi(1+2l-l^2)}{(1+l)^2 \sqrt{1+2l}}
\sqrt{x^2-\frac{2(1+l)}{1+2l}z}}} \Bigg| \leq  \frac{C \sqrt{-z}}
{\sqrt[4]{(x^2-z)^3}}+C\sqrt{|x|}.
$$
It follows from the last estimate and Lemma~\ref{decomp of Delta} that
\begin{eqnarray*}
|F_2(z;x,y)|&&\leq \frac{C \sqrt{-z}} {\sqrt[4]{(x^2-z)^3}}  \frac{|xy|}
{(1+l)x^2+2lxy+(1+l)y^2-2z} \frac{1}{\sqrt[4]{y^2-z}}\\
&&+ \frac{C|x|^{3/2}|y|}
{(1+l)x^2+2lxy+(1+l)y^2-2z} \frac{1}{\sqrt[4]{y^2-z}}
\end{eqnarray*}
or
$$
|F_2(z;x,y)|\leq \frac{C|x|^{1/2}|y|^{3/2} \sqrt{-z}}
{(1+l)x^2+2lxy+(1+l)y^2-2z}+\frac{C|x|^{3/2}|y|^{3/2}}{(1+l)x^2+2lxy+(1+l)y^2-2z}
$$
for all $x,y \in (-\delta, \delta)$ with some positive constant $C.$
Since
$$
\int_{-\delta}^\delta\int_{-\delta}^\delta \Big(\frac{|x|^{1/2}|y|^{3/2}}
{(1+l)x^2+2lxy+(1+l)y^2-2z}\Big)^2 dxdy\leq C |\log(-z)|,
$$
for any fixed $z \leq 0$ the function $F_2(z;\cdot,\cdot)$ is a square-integrable on ${\Bbb T}^2.$
By the same way we can show the square-integrability of
$F_3(z;\cdot,\cdot)$ on ${\Bbb T}^2$ for any fixed $z \leq 0.$

Hence, the operator $T_{11}(z)-T_{11}(\delta; z)$ belongs to the Hilbert-Schmidt class for
all $z\leq 0.$ In combination with the continuity of the kernel of
the operator with respect to $z<0,$ this implies the continuity of
$T_{11}(z)-T_{11}(\delta; z)$ with respect to $z \leq 0.$

By the definition the operators $T_{00}(z),$
$T_{01}(z)$ and $T_{01}^*(z)$ are rank 1 operators and they
are continuous from the left up to $z=0.$ Consequently
the operator $F(z)$ is compact and the operator-valued function
$F(\cdot)$ is continuous in the operator-norm in $(-\infty, 0].$
\end{proof}

By the structure of $T(\delta; z)$ we have $\sigma(T(\delta;
z))=\{0\} \cup \sigma(T_{11}(\delta; z)).$

The subspace of functions $g$ having support in $(-\delta, \delta)$
is an invariant subspace for the operator $T_{11}(\delta; z).$ Let
$T_{11}^{(0)}(\delta; z)$ be the restriction of the operator
$T_{11}(\delta; z)$ to the subspace $L_2(-\delta, \delta),$ that is,
the integral operator with the kernel
$$
\frac{1}{\pi} \frac{(1+l)^2 \sqrt{1+2l}}{1+2l-l^2}
\frac{1}{\sqrt[4]{x^2-\frac{2(1+l)}{2l+1}z }} \frac{xy}{(1+l)x^2+2lxy+(1+l)y^2-2z}
\frac{1}{\sqrt[4]{ y^2-\frac{2(1+l)}{2l+1}z}},
$$
$x,y\in (-\delta, \delta).$
Then we have $\sigma(T_{11}(\delta; z))=
\sigma(T_{11}^{(0)}(\delta; z)),$ $z<0.$

Let $L_2^{\rm o}(-\delta, \delta)$ and $L_2^{\rm e}(-\delta, \delta)$ be the
spaces of odd and even functions, respectively. It is easily to
check that $T_{11}^{(0)}(\delta; z): L_2^{\rm o}(-\delta, \delta) \to
L_2^{\rm o}(-\delta, \delta)$ and $T_{11}^{(0)}(\delta; z):
L_2^{\rm e}(-\delta, \delta)\to L_2^{\rm e}(-\delta, \delta).$

Let us consider the unitary operator
$$
U_{\rm o}: L_2^{\rm o}(-\delta, \delta) \to L_2(0, \delta),\quad
(U_{\rm o}f)(x)= \sqrt{2}f(x).
$$
Then
$$
U_{\rm o}^{-1}: L_2(0, \delta) \to L_2^{\rm o}(-\delta, \delta),\quad
(U_{\rm o}^{-1}f)(x)=\left \{
\begin{array}{ll}
\frac{1}{\sqrt{2}}f(x) & \mbox{as}\quad x\geq 0\\
-\frac{1}{\sqrt{2}}f(x) & \mbox{as}\quad x<0.
\end{array}\right.
$$
Let $T_{\rm o}(z):=U_{\rm o}T_{11}^{(0)}(\delta; z)U_{\rm o}^{-1}.$ Then
$\sigma(T_{11}^{(0)}(\delta; z)) \supset \sigma(T_{\rm o}(z)),$ where
$T_{\rm o}(z)$ is the integral operator acting on $L_2(0, \delta)$ with
the kernel
\begin{eqnarray*}
T_{\rm o}(z;x,y):&&=\frac{1}{\pi} \frac{(1+l)^2 \sqrt{1+2l}}{1+2l-l^2}
\frac{1}{\sqrt[4]{x^2-\frac{2(1+l)}{2l+1}z }} \Big
[\frac{xy}{(1+l)x^2+2lxy+(1+l)y^2-2z}\\
&&+ \frac{xy}{(1+l)x^2-2lxy+(1+l)y^2-2z}\Big ] \frac{1}{\sqrt[4]{
y^2-\frac{2(1+l)}{2l+1}z}}.
\end{eqnarray*}

Let $T_1(z),$ $z<0$ be an integral operator on  $L_2(0, \delta)$ with
the kernel $\chi_{\Omega(z)}(x)T_{\rm o}(z;x,y)\chi_{\Omega(z)}(y),$ where $\Omega(z):=(|z|^{1/2},\delta].$

\begin{lemma}\label{T(o) compact}
Let $\mu=\mu_0.$ Then for any $z\in (-\delta, 0]$ the operator
$G(z):=T_{\rm o}(z)-T_3(z)$ is compact and the operator-valued function
$G(\cdot)$ is continuous in the operator-norm in $(-\delta, 0].$
\end{lemma}

The subspace of functions $g$ having support in $\Omega(z)$ is an
invariant subspace for the operator $T_1(z).$ Let $T_2(z)$ be the
restriction of the operator $T_1(z)$ to the subspace
$L_2(\Omega(z)),$ that is, the integral operator with kernel
$T_{2}(z;x,y):=T_{\rm o}(z;x,y),$ $x,y\in \Omega(z).$
Then  we have $\sigma(T_{1}(z))= \sigma(T_2(z))$ for $z\in (-\delta, 0].$

Let us consider the unitary dilation
\begin{eqnarray*}
&&U: L_2(\Omega(z))\to L_2(-\pi,\pi),\\
&&(Uf)(x)=\sqrt{\frac{R(z)}{2\pi}}
(|z|^{1/2}e^{\frac{R(z)}{2\pi}(x+\pi)})^{1/2}f(|z|^{1/2} e^{\frac{R(z)}{2\pi}(x+\pi)}),\\
&&U^{-1}: L_2(-\pi,\pi)\to L_2(\Omega(z)),\\
&&(U^{-1}f)(p)=\sqrt{\frac{2\pi}{R(z)}}|x|^{-1/2}\,f(
\frac{2\pi}{R(z)}\log\frac{|x|}{|z|^{1/2}}-\pi), \quad R(z):=-
\log\frac{|z|^{1/2}}{\delta}.
\end{eqnarray*}

The operator $T_3(z):=UT_2(z)U^{-1}$ is integral operator on
$L_2(-\pi,\pi)$ with the kernel
\begin{eqnarray*}
&& T_3(z;x,y):= \frac{1}{\pi}
\frac{(1+l)^2 \sqrt{1+2l}}{1+2l-l^2}\frac{ R(z)}{2\pi}
\frac{e^{\frac{3R(z)}{4\pi}(x+\pi)}}
{(e^{\frac{R(z)}{\pi}(x+\pi)}+\frac{2(1+l)}{2l+1})^{1/4}}\\
&& \times \Big[
\frac{1}{(1+l)e^{\frac{R(z)}{\pi}(x+\pi)}+2le^{\frac{R(z)}{2\pi}(x+\pi)}e^{\frac{R(z)}{2\pi}(y+\pi)}+(1+l)e^{\frac{R(z)}{\pi}(y+\pi)}+2}\\
&&+\frac{1}{(1+l)e^{\frac{R(z)}{\pi}(x+\pi)}-2le^{\frac{R(z)}{2\pi}(x+\pi)}e^{\frac{R(z)}{2\pi}(y+\pi)}+(1+l)e^{\frac{R(z)}{\pi}(y+\pi)}+2}\Big]
\frac{e^{\frac{3R(z)}{4\pi}(y+\pi)}}
{(e^{\frac{R(z)}{4\pi}(y+\pi)}+\frac{2(1+l)}{2l+1})^{1/4}}.
\end{eqnarray*}

\begin{lemma}\label{T compact}
Let $\mu=\mu_0.$ Then for any $z \in (-\delta, 0]$ the operator
$G_1(z):=T_3(z)-T_4(z)$ is compact and the operator-valued
function $G_1(\cdot)$ is continuous in the operator-norm in $(-\delta, 0],$
where the operator $T_4(z)$ is an integral operator on $L_2(-\pi,\pi)$
with kernel $T_4(z;x),$
$$
T_4(z;x):= \frac{1}{\pi}
\frac{(1+l)^2 \sqrt{1+2l}}{1+2l-l^2}\frac{ R(z)}{2\pi}
\Big[\frac{1}{2(1+l) {\rm ch}\,\frac{R(z)}{2\pi}(x)+2l}+\frac{1}{2(1+l)
{\rm ch}\,\frac{R(z)}{2\pi}(x)-2l}\Big].
$$
\end{lemma}

Let us define in $L_2(-\pi, \pi)$ the operator $S(z),$ $z\in (-\delta, 0)$ by
\begin{eqnarray*}
&& S(z):=\sum\limits_{k\in \mathbb{Z}}\lambda_k(z)
(\varphi_k,\cdot)\varphi_k,\\
&& \lambda_k(z):=\frac{(1+l)^2 \sqrt{1+2l}}{1+2l-l^2} \frac{1}{1+l}\frac{1}{\sin
(\arccos \frac{l}{1+l})} \frac{{\rm sh}(\arccos
\frac{l}{1+l}\frac{2k\pi}{R(z)})+{\rm sh}((\pi-\arccos
\frac{l}{1+l})\frac{2k\pi}{R(z)})}{{\rm sh} \frac{2k\pi^2}{R(z)}},
\end{eqnarray*}
where
$\varphi_{0}(x):=\frac{1}{2\pi}$ and
$\varphi_{n}(x):=\frac{1}{\sqrt{2\pi}}e^{i nx}$ as $n\neq 0.$

\begin{lemma}\label{S compact}
Let $\mu=\mu_0.$ Then for any $z \in (-\delta, 0]$ the operator
$G_2(z):=T_4(z)-S(z)$ is compact and the operator-valued function
$G_2(\cdot)$ is continuous in the operator-norm in $(-\delta, 0].$
\end{lemma}

\begin{proof}
Note that the operator $T_4(z)$ is convolution type.
Therefore the eigenvalues of $T_4(z)$ can be found. By the Hilbert-Schmidt theorem
the operator $T_4(z)$ can be decomposed as
$$
T_4(z) =\sum\limits_{n\in \mathbb{Z}}u_n(z)
(\varphi_{n},\cdot)\varphi_{n},
$$
where
$$
u_k(z):=
\frac{(1+l)^2 \sqrt{1+2l}}{1+2l-l^2} \frac{R(z)}{2\pi^2}
\int_{-\pi}^\pi \Big[\frac{e^{i nt}}{2 (1+l){\rm ch}(\frac{
R(z)}{2\pi}t)+2l}+\frac{e^{i nt}}{2 (1+l){\rm ch}(\frac{
R(z)}{2\pi}t)-2l} \Big]dt.
$$
We represent $u_k(z)$ as
\begin{eqnarray}
u_k(z)=\tilde u_k(z)-O_k(z),
\label{representation for ukz}
\end{eqnarray}
where
\begin{eqnarray*}
&& \tilde u_k(z):=\frac{(1+l)^2 \sqrt{1+2l}}{1+2l-l^2} \frac{1}{\pi}
\int_{-\infty}^\infty \Big[\frac{e^{i\frac{2\pi}{ R(z)}nt}}{2
(1+l) {\rm ch} s+2l}+\frac{e^{i\frac{2\pi}{ R(z)}nt}}{2 (1+l) {\rm ch} s-2l}
\Big]dt,\\
&& O_k(z):=\frac{(1+l)^2 \sqrt{1+2l}}{1+2l-l^2} \frac{1}{\pi}
\int_{|t|>\frac{R(z)}{2}} \Big[\frac{e^{i\frac{2\pi}{
R(z)}nt}}{2 (1+l) {\rm ch} s+2l}+\frac{e^{i\frac{2\pi}{ R(z)}nt}}{2
(1+l) {\rm ch} s-2l} \Big]dt.
\end{eqnarray*}
Using the equality
$$
\frac{1}{2\pi}\int_{-\infty}^{\infty}\frac{e^{it r}\,\sin a}
{{\rm ch}\,t +\cos a}dt=\frac{{\rm sh}(a r)}{{\rm sh}(\pi r)},\quad \mbox{as}\quad
|a|<\pi
$$
we get the equality $\tilde u_k(z)=\lambda_k(z).$
It follows from (\ref{representation for ukz}) that
the  difference $G_2(z):=T_4(z)-S(z)$  is a Hilbert-Schmidt operator
and continuous up to $z =0.$
\end{proof}

\begin{lemma}\label{n(1,S)}
Let $\mu=\mu_0.$ There exist $l>0$ and  $\rho>0$ such that
$\lim\limits_{z\to -0}n(1+\rho, S(z))=\infty.$
\end{lemma}

\begin{proof}
Since
$$
\lambda_k(0):=\lim\limits_{z\to -0}\lambda_k(z)=
\frac{(1+l)^2 \sqrt{1+2l}}{1+2l-l^2} \frac{1}{1+l}\frac{1}{\sin
(\arccos \frac{l}{1+l})} \frac{\arccos \frac{l}{1+l}+{\rm sh}(\pi-\arccos
\frac{l}{1+l})}{\pi}.
$$
It is easy to check that for any $k\in \mathbb{Z}$ it takes place $\lambda_k(0)>1$ as $l=2.$
\end{proof}

Main result of this section is the following statement.
\begin{theorem}\label{inf of eigen H}
Let $\mu=\mu_0$ and the parameter functions $v_1(\cdot),$
$w_1(\cdot)$ and $w_2(\cdot,\cdot)$ be given by \eqref{parametrs H}. Then there exists a $l>0$ such that
the operator $H$ has a infinite number of eigenvalues lying below
$E_{\rm min}=0.$
\end{theorem}

\begin{proof}
Using the Weyl's inequality (\ref{Weyl ineq}) for any $z \in (-\delta, 0]$ we have  the
inequalities
\begin{eqnarray*}
&& n(1+\delta,S(z))\leq n(1+\frac{4\delta}{5},T_4(z))+n(\frac{\delta}{5},S(z)-T_4(z)),\\
&& n(1+\frac{4\delta}{5},T_4(z))\leq n(1+\frac{3\delta}{5},T_3(z))+n(\frac{\delta}{5},T_4(z)-T_3(z)),\\
&& n(1+\frac{3\delta}{5},T_3(z))=n(1+\frac{3\delta}{5},T_1(z))\leq n(1+\frac{2\delta}{5},T_1(z))+n(\frac{\delta}{5},T_1(z)-T_{\rm o}(z)),\\
&& n(1+\frac{2\delta}{5},T_o(z))=n(1+\frac{\delta}{5},T_{11}^{(0)}(\delta;z))+n(\frac{\delta}{5},T_{\rm o}(z)-T_{11}^{(0)}(\delta;z)),\\
&& n(1+\frac{\delta}{5},T_{11}^{(0)}(\delta;z))=n(1+\frac{\delta}{5},T(\delta;z))\leq n(1,T(z))+n(\frac{\delta}{5},T(\delta;z)-T(z)).
\end{eqnarray*}
According to Lemmas~\ref{T(0) compact}$-$\ref{S compact} we get the inequalities
\begin{eqnarray*}
&& n(\frac{\delta}{5},S(z)-T_4(z))<\infty
,\,n(\frac{\delta}{5},T_4(z)-T_3(z))<\infty,\,
n(\frac{\delta}{5},T_1(z)-T_{\rm o}(z))<\infty,\\
&& n(\frac{\delta}{5},T_{\rm o}(z)-T_{11}^{(0)}(\delta;z))<\infty,\quad
n(\frac{\delta}{5},T(\delta;z)-T(z))<\infty
\end{eqnarray*}
for any $z \in (-\delta, 0].$ Then
$
n(1+\rho,S(z))\leq C+ n(1,T(z)),
$
where $C>0$ does not depend of $z \in
(-\delta, 0].$ Hence by Lemma~\ref{n(1,S)} we obtain the proof of the
theorem.
\end{proof}

{\bf Case II: Finiteness.} Let the parameter functions $v_1(\cdot),$ $w_1(\cdot)$ and $w_2(\cdot,\cdot)$
have the form
\begin{eqnarray*}
 v_1(x)=\sqrt{\mu} (1-\cos x),\quad \mu>0;\quad
 w_1(x)=2-\cos x;\quad
 w_2(x,y)=2-\cos x-\cos y.
\end{eqnarray*}

Then the function $w_2(\cdot,\cdot)$ has a unique non-degenerate global zero
minimum ($m=0$) at the point $(0,0) \in {\Bbb T}^2$
and $v_1(0)=0.$ It is easy to see that for
$$
\Delta(x\,; z)=2-\cos x-z-\mu \int_0^\pi \frac{(1-\cos t)^2\, dt}{2-\cos x-\cos t-z}
$$
we have $\Delta(0\,; 0)=0$ if and only if $\mu=1/\pi.$

\begin{lemma}\label{decomp of Delta1}
Let $\mu=1/\pi.$ Then there exist the numbers $C_1, C_2>0$ and $\delta>0$ such that
$$
C_1 x^2 \leq \Delta(x\,; 0) \leq C_2 x^2, \quad x \in U_\delta(0).
$$
\end{lemma}

\begin{proof}
Let $\delta>0$ be sufficiently small. We rewrite the function $\Delta(\cdot\,; 0)$
in the form
$
\Delta(x\,; 0)=\Delta_1(x)+\Delta_2(x),
$
where
\begin{eqnarray*}
 \Delta_1(x):=2-\cos x-\mu \int_\delta^\pi \frac{(1-\cos t)^2\, dt}{2-\cos x-\cos t-z},\,\,
 \Delta_2(x\,; z):=-\mu \int_0^\delta \frac{(1-\cos t)^2\, dt}{2-\cos x-\cos t-z}.
\end{eqnarray*}
Since $\Delta_1(\cdot)$ is an even analytic function on ${\Bbb T},$ we have
\begin{eqnarray}
\Delta_1(x)=\Delta_1(0)+O(x^2)
\label{decomp for Delta11}
\end{eqnarray}
as $x \to 0.$
Using the expansion (\ref{decomp for sin and cos}) for $1-\cos x$
we obtain
$$
\Delta_2(x):=-\frac{\mu}{4} \int_0^\delta \frac{t^4 dt}{x^2+t^2}+O(x^2)
$$
as $x \to 0.$
Since
$$
\int_0^\delta \frac{t^4 dt}{x^2+t^2}=\frac{\delta^3}{3}-\delta x^2+x^4\int_0^\delta \frac{dt}{x^2+t^2},
$$
by the properties (\ref{arctan 1}) and (\ref{arctan 2}) we have
\begin{eqnarray}
\Delta_2(x)=\Delta_2(0)+O(x^2)
\label{decomp for Delta12}
\end{eqnarray}
as $x \to 0.$ Recall that if $\mu=1/\pi,$ then $\Delta(0\,; 0)=0.$
Now, taking into account the equalities (\ref{decomp for Delta11}) and (\ref{decomp for Delta12})
we obtain $\Delta(x\,; 0)=O(x^2)$ as $x \to 0,$ which implies that there exist $C_1, C_2>0$ and $\delta>0$
such that the assertion of lemma holds.
\end{proof}

\begin{lemma}\label{T(z) compact}
Let $\mu=1/\pi.$ For any $z \leq 0$ the operator $T(z)$ is
compact and continuous on the left up to $z=0.$
\end{lemma}

\begin{proof}
Let $\mu=1/\pi.$ Denote by $Q(z;x,y)$ the kernel of the
integral operator $T_{11}(z),$ $z<0,$ that is,
$$
Q(z;x,y):=\frac{v_1(x)v_1(y)}{2\sqrt{\Delta(x\,; z)}(w_2(x,y)-z) \sqrt{\Delta(y\,; z)}}.
$$

By  virtue of decomposition
(\ref{decomp for sin and cos}) and Lemma~\ref{decomp of Delta1} the kernel $Q(z;x,y)$ is estimated by the
square-integrable function
$$
C_1 \left(1+\frac{\chi_{\delta}(x) \chi_{\delta}(y) |x| |y|}{x^2+y^2} \right),
$$
defined on ${\Bbb T}^2,$ where $\chi_\delta(\cdot)$ is the characteristic function of
$(-\delta, \delta).$ Hence for
any $z\leq 0$ the operator $T_{11}(z)$ is Hilbert-Schmidt.

The kernel function of $T_{11}(z),$ $z<0$ is
continuous in $x,y\in {\Bbb T}.$ Therefore the continuity of the
operator $T_{11}(z)$ from the left up to $z=0$ follows
from Lebesgue's dominated convergence theorem.
Since for all $z \leq 0$ the operators $T_{00}(z),$
$T_{01}(z)$ and $T_{01}^*(z)$ are of rank
1 and continuous from the left up to $z=0$ one concludes that
$T(z)$ is compact and continuous from the left up
to $z=0.$
\end{proof}

Using Lemma \ref{T(z) compact} we can now proceed analogously to the proof of Theorem~\ref{THM 3}
to show the finiteness of the negative discrete spectrum of $H.$

\section{Application}

In this section we investigate the spectrum of ${\cal A}_2,$
introduced in Section~1 applying the results for $H.$ We recall that
the operator ${\cal A}_2$ has a $3 \times 3$ tridiagonal block operator matrix representation
$$
{\cal A}_2:=\left( \begin{array}{ccc}
{\cal A}_{00} & {\cal A}_{01} & 0\\
{\cal A}_{01}^* & {\cal A}_{11} & {\cal A}_{12}\\
0 & {\cal A}_{12}^* & {\cal A}_{22}\\
\end{array}
\right),
$$
where matrix elements $A_{ij}$ are defined by
\begin{eqnarray*}
&& {\cal A}_{00}f_0^{(\sigma)}=\varepsilon \sigma f_0^{(\sigma)}, \quad {\cal A}_{01}f_1^{(\sigma)}=\alpha \int_{\Bbb T}
v(t)f_1^{(-\sigma)}(t)dt,\\
&& ({\cal A}_{11}f_1^{(\sigma)})(x)=(\varepsilon \sigma+w(x))f_1^{(\sigma)}(x),\quad
({\cal A}_{12}f_2^{(\sigma)})(x)= \alpha \int_{\Bbb T} v(t) f_2^{(-\sigma)}(x,t)dt,\\
&& ({\cal A}_{22}f_2^{(\sigma)})(x,y)=(\varepsilon \sigma+w(x)+w(y))f_2^{(\sigma)}(x,y), \quad
f=\{f_0^{(\sigma)},f_1^{(\sigma)},f_2^{(\sigma)}; \sigma=\pm\} \in {\mathcal L}_2.
\end{eqnarray*}

We make the following assumptions: $\varepsilon>0;$ the dispersion $w(\cdot)$ is an
analytic on ${\Bbb T}$ and has a unique zero minimum at the point $0 \in {\Bbb T};$
$v(\cdot)$ is a real-valued analytic function on ${\Bbb T};$ the coupling constant
$\alpha>0$ is an arbitrary.

Consider the following permutation operator
\begin{eqnarray*}
&&\Phi: {\mathcal L}_2 \to {\cal H} \oplus {\cal H}, \\
&&\Phi: (f_0^{(+)}, f_0^{(-)}, f_1^{(+)}, f_1^{(-)}, f_2^{(+)}, f_2^{(-)}) \to
(f_0^{(+)}, f_1^{(-)}, f_2^{(+)}, f_0^{(-)}, f_1^{(+)}, f_2^{(-)}).
\end{eqnarray*}

To investigate the spectral properties of ${\cal A}_2$ we introduce
the following two bounded self-adjoint operators ${\cal A}_2^{(\sigma)},$ $\sigma=\pm,$ which acts in ${\mathcal F}_{\rm s}^{(2)}(L_2({\Bbb T}))$ as
$$
{\cal A}_2^{(\sigma)}:=\left( \begin{array}{ccc}
\widehat{{\cal A}}_{00}^{(\sigma)} & \widehat{{\cal A}}_{01} & 0\\
\widehat{{\cal A}}_{01}^* & \widehat{{\cal A}}_{11}^{(\sigma)} & \widehat{{\cal A}}_{12}\\
0 & \widehat{{\cal A}}_{12}^* & \widehat{{\cal A}}_{22}^{(\sigma)}\\
\end{array}
\right)
$$
with the entries
\begin{eqnarray*}
&& \widehat{{\cal A}}_{00}^{(\sigma)}f_0=\varepsilon \sigma f_0, \quad \widehat{{\cal A}}_{01}f_1=\alpha \int_{\Bbb T}
v(t)f_1(t)dt,\\
&& (\widehat{{\cal A}}_{11}^{(\sigma)}f_1)(x)=(-\varepsilon \sigma+w(x))f_1(x),\quad
(\widehat{{\cal A}}_{12}f_2)(x)= \alpha \int_{\Bbb T} v(t) f_2(x,t)dt,\\
&& (\widehat{{\cal A}}_{22}^{(\sigma)}f_2)(x,y)=(\varepsilon \sigma+w(x)+w(y))f_2(x,y),\quad
(f_0,f_1,f_2) \in {\mathcal F}_{\rm s}^{(2)}(L_2({\Bbb T})).
\end{eqnarray*}

The definitions of the operators ${\cal A}_2,$ ${\cal A}_2^{(\sigma)}$ and $\Phi$ imply that
$$
\Phi {\cal A}_2 \Phi^{-1}={\rm diag} \{{\cal A}_2^{(+)}, {\cal A}_2^{(-)}\}.
$$

The following theorem describes the relation between spectra of ${\cal A}_2$
and ${\cal A}_2^{(\sigma)}.$

\begin{theorem}\label{spectra of N=2}
The equality $\sigma({\cal A}_2)=\sigma({\cal A}_2^{(+)}) \cup \sigma({\cal A}_2^{(-)})$
holds. Moreover,
\begin{eqnarray*}
\sigma_{\rm ess}({\cal A}_2)=\sigma_{\rm ess}({\cal A}_2^{(+)}) \cup \sigma_{\rm ess}({\cal A}_2^{(-)}),\quad
\sigma_{\rm p}({\cal A}_2)=\sigma_{\rm p}({\cal A}_2^{(+)}) \cup \sigma_{\rm p}({\cal A}_2^{(-)}).
\end{eqnarray*}
\end{theorem}

\begin{remark}
Since the part of $\sigma_{\rm disc}({\cal A}_2^{(\sigma)})$ can be located in $\sigma_{\rm ess}({\cal A}_2)$ we have the inclusion
\begin{eqnarray}
\sigma_{\rm disc}({\cal A}_2) \subseteq \sigma_{\rm disc}({\cal A}_2^{(+)}) \cup \sigma_{\rm disc}({\cal A}_2^{(-)}).
\label{inclusion}
\end{eqnarray}
\end{remark}

To describe the location of the essential spectrum of ${\cal A}_2$
we introduce the following two families of bounded self-adjoint operators  $h^{(\sigma)}(x),$ $x \in {\Bbb T},$
which acts in ${\mathcal F}_{\rm s}^{(1)}(L_2({\Bbb T}))$ as
$$
h^{(\sigma)}(x):=\left( \begin{array}{cc}
h_{00}^{(\sigma)}(x) & h_{01}\\
h_{01}^* & h_{11}^{(\sigma)}(x)\\
\end{array}
\right),
$$
where
\begin{eqnarray*}
&&h_{00}^{(\sigma)}(x)f_0=(-\sigma\varepsilon+w(x))f_0,\quad h_{01}f_1=\frac{\alpha}{\sqrt{2}}
\int_{\Bbb T} v(t) f_1(t)dt,\\
&&(h_{11}^{(\sigma)}(x)f_1)(y)=(\sigma\varepsilon+w(x)+w(y))f_1(y),\quad (f_0,f_1) \in {\mathcal F}_{\rm s}^{(1)}(L_2({\Bbb T})).
\end{eqnarray*}


By Theorem \ref{THM 1} for the essential spectrum of ${\cal A}_2^{(\sigma)}$
the following equality holds
$$
\sigma_{\rm ess}({\cal A}_2^{(\sigma)})=\bigcup\limits_{x \in {\Bbb T}} \sigma_{\rm disc}(h^{(\sigma)}(x))
\cup [\sigma\varepsilon, 2M_w+\sigma\varepsilon],\quad
M_w:= \max\limits_{x\in {\Bbb T}} w(x).
$$
Now taking into account last equality we obtain from Theorem~\ref{spectra of N=2} that
$$
\sigma_{\rm ess}({\cal A}_2)=\bigcup\limits_{\sigma=\pm}\bigcup\limits_{x \in {\Bbb T}} \sigma_{\rm disc}(h^{(\sigma)}(x))
\cup [-\varepsilon, 2M_w-\varepsilon] \cup [\varepsilon, 2M_w+\varepsilon].
$$

To estimate the lower bound of the essential spectrum of ${\cal A}_2,$
for any $x\in {\Bbb T}$ we define the Fredholm determinant $\Delta^{(\sigma)}(x\,; \cdot):$
$$
\Delta^{(\sigma)}(x\,; z):=-\sigma\varepsilon+w(x)-z-\frac{\alpha^2}{2} \int_{\Bbb T}
\frac{v^2(t)dt}{\sigma\varepsilon+w(x)+w(t)-z}
$$
in ${\Bbb C} \setminus [\sigma\varepsilon+w(x), M_w+\sigma\varepsilon+w(x)],$ associated with the
operator $h^{(\sigma)}(x).$
By Lemma~\ref{LEM 1} for the discrete spectrum of $h^{(\sigma)}(x)$ the equality
$$
\sigma_{\rm disc}(h^{(\sigma)}(x))=\{z \in {\Bbb C} \setminus [\sigma\varepsilon+w(x), M_w+\sigma\varepsilon+w(x)]:\,
\Delta^{(\sigma)}(x\,; z)=0 \}
$$
holds. By definition of $\Delta^{(\sigma)}(\cdot\,; \cdot)$ we have
$$
\min\limits_{x \in {\Bbb T}}\Delta^{(\sigma)}(x\,; z)=-\varepsilon\sigma-z-\frac{\alpha^2}{2} \int_{\Bbb T}
\frac{v^2(t)dt}{\sigma\varepsilon+w(t)-z}, \quad z \leq \sigma\varepsilon.
$$
If we set $E_{\rm min}^{(\sigma)}:=\min \sigma_{\rm ess}({\cal A}_2^{(\sigma)}),$
then Theorem~\ref{spectra of N=2} implies
$$
E_{\rm min}:=\min \sigma_{\rm ess}({\cal A}_2)=\min \{E_{\rm min}^{(-)}, E_{\rm min}^{(+)}\}.
$$
It is clear that $\min\limits_{x \in {\Bbb T}}\Delta^{(+)}(x\,; -\varepsilon)<0$ for all $\alpha>0$
and hence $E_{\min} \leq E_{\min}^{(+)}<-\varepsilon.$

Now we study the lower bound of the essential spectrum of ${\cal A}_2^{(-)}.$ If
$$
\int_{\Bbb T}\frac{v^2(t)dt}{w(t)}=\infty,
$$
then again $\min\limits_{x \in {\Bbb T}}\Delta^{(-)}(x\,; -\varepsilon)<0$ for all $\alpha>0,$ that is,
$E_{\min} \leq E_{\min}^{(-)}<-\varepsilon.$ If
$$
\int_{\Bbb T}\frac{v^2(t)dt}{w(t)}<\infty,
$$
then
\begin{eqnarray*}
&&\min\limits_{x \in {\Bbb T}}\Delta^{(-)}(x\,; -\varepsilon)<0 \Leftrightarrow \alpha>\alpha_0:=2\sqrt{\varepsilon}
\Bigl( \int_{\Bbb T}\frac{v^2(t)dt}{w(t)} \Bigr)^{-1/2};\\
&&\min\limits_{x \in {\Bbb T}}\Delta^{(-)}(x\,; -\varepsilon) \geq 0 \Leftrightarrow \alpha \leq \alpha_0.
\end{eqnarray*}
So, by Theorem~\ref{THM 2} we obtain $E_{\rm min}^{(-)}<-\varepsilon$ for all $\alpha>\alpha_0$
and $E_{\rm min}^{(-)}=-\varepsilon$ for all $\alpha \leq \alpha_0.$

The above analysis leads to $E_{\rm min}<-\varepsilon$ for all $\alpha>0.$

Since the parameter functions of ${\cal A}_2^{(+)}$ satisfy the conditions of Theorem~\ref{THM 3},
it has finitely many eigenvalues smaller than $E_{\rm min}^{(+)}$ for all $\alpha>0.$
Similarly for any $\alpha>\alpha_0$ the operator ${\cal A}_2^{(-)}$ has a finitely many eigenvalues smaller than $E_{\rm min}^{(-)}.$
For the case $\alpha \leq \alpha_0$ we have $E_{\rm min}^{(-)}=-\varepsilon$ and $E_{\rm min}<E_{\rm min}^{(-)}.$ Hence
the operator ${\cal A}_2^{(-)}$ has a finitely many eigenvalues smaller than $E_{\rm min}.$
Now by the inclusion (\ref{inclusion}) we conclude that
the operator ${\cal A}_2$ has a finitely many eigenvalues smaller than $E_{\rm min}.$

\section*{Acknowledgment}
This work was supported by the IMU Einstein Foundation Program.
T. H. Rasulov wishes to thank the Berlin Mathematical School and Weierstrass Institute
for Applied Analysis and Stochastics for the invitation and hospitality.

\end{document}